\documentclass[12pt]{article}
\usepackage{amssymb,euler,amsmath,defns}
\usepackage{url,color,rotating,hyperref}
\allowdisplaybreaks
\newcommand{\gamd}{\gamma'}
\newcommand{\cL}{\mathcal L}
\newcommand{\inpr}[2]{\left<#1,#2\right>}
\newcommand{\shft}[1]{\varepsilon_{#1}}
\newcommand{\bshft}[1]{\bar \varepsilon_{#1}}

\newcommand{\bmu}{\bar\mu}
\newcommand{\bdelta}{\bar\delta}

\newcommand{\Kappa}{{\mathcal K}}
\newcommand{\cE}{{\mathbb E}}
\newcommand{\cM}{{\mathcal M}}
\newcommand{\diag}{\operatorname{diag}}
\newcounter{rgb}

\title{Choose interelement coupling to preserve self-adjoint dynamics in multiscale modelling and computation}

\author{A.~J. Roberts\thanks{School of Mathematical Sciences, University of Adelaide, South Australia~5005.  
\protect\url{mailto:anthony.roberts@adelaide.edu.au}}}

\begin{document}
\maketitle

\begin{abstract}
Consider the macroscale modelling of microscale spatiotemporal dynamics.
Here we develop a new approach to ensure coarse scale discrete models preserve important self-adjoint properties of the fine scale dynamics.  The first part explores the discretisation of microscale continuum dynamics.   The second addresses how dynamics on a fine lattice are mapped to lattice a factor of two coarser (as in multigrids).  Such mapping of discrete lattice dynamics may be iterated to empower us in future research to explore scale dependent emergent phenomena.  The support of dynamical systems, centre manifold, theory ensures that the coarse scale modelling applies with a finite spectral gap, in a finite domain, and for all time.  The accuracy of the models is limited by the asymptotic resolution of subgrid coarse scale processes, and is controlled by the level of truncation.  As given examples demonstrate, the novel feature of the approach developed here is that it ensures the preservation of important conservation properties of the microscale dynamics. 
\end{abstract}

\tableofcontents

\section{Introduction}

Dynamical systems theory gives \emph{new} assurances about the quality of finite difference and finite element models of nonlinear spatiotemporal systems.  Efforts to construct approximations to the  long-term, low-dimensional dynamics of dissipative partial differential equations (\pde{}s), on its inertial manifold~\cite[e.g.]{Temam90}, have largely been based upon the \emph{global} nonlinear Galerkin  method~\cite[e.g.]{Roberts89, Foias88b, Marion89}, and its variants~\cite[e.g.]{Jolly90, Foias91b}. In contrast, a `holistic discretisation'~\cite{Roberts98a} developed further here is based purely upon the \emph{local} dynamics on finite elements while maintaining, as do inertial manifolds, fidelity with the solutions of the original \pde.

To generate a macroscale model of a dissipative \pde, we divide space into finite elements, then specially crafted coupling conditions empower us to support a macroscale model with centre manifold theory~\cite[e.g.]{Carr81, Kuznetsov95}.  Such macroscale models are used for computations and for analytic understanding of the dynamics.  Crucially, the theory supports the existence, exponentially quick attractiveness, and approximate construction of a slow manifold of the system dynamics, both deterministic~\cite[e.g.]{Roberts98a} and stochastic~\cite{Roberts08a}.  Earlier research developed coupling conditions that \emph{both} had centre manifold support on finite elements \emph{and} ensured consistency as the element size became small~\cite[e.g.]{Roberts00a, MacKenzie05a}.  However, albeit effective for many systems, these coupling conditions fail to preserve \emph{at each level of approximation} any self-adjoint symmetry in the underlying spatial dynamics~\cite{Roberts01d}.  The innovations introduced in Section~\ref{sec:sacc} are interelement coupling conditions that not only engender centre manifold support, Section~\ref{sec:cmtsd}, and assure consistency for vanishing element size, but also preserve self-adjoint symmetry \emph{in each approximation}.  For just one example, Section~\ref{sec:nd}, as deduced in the model~\eqref{eq:nded}, argues that a particular one-dimensional, nonlinear, continuum diffusion is soundly mapped as follows to nonlinear dynamics on a macroscale grid (spacing~$h$):
\begin{equation}
\D tu=\D x{}\left[u\D xu\right]
\quad\mapsto\quad
\frac{dU_j}{dt}\approx\frac1{2h^2}\left( U_{j+1}^2-2U_j^2+U_{j+1}^2 \right)\,,
\label{eq:introdiff}
\end{equation}
where the grid values $U_j=u(X_j,t)$. Why this discretisation instead of others?  Because the new coupling conditions as well as having dynamical systems support and classic consistency, also preserve in the model the self-adjoint symmetry of material conservation that is present in the original \pde.

In many science and engineering applications it is essential to preserve the conservative form of the governing equations. I argue that interelement coupling rules akin to those of Section~\ref{sec:sacc} will \emph{automatically} preserve conservative forms.

Most methods for modelling dynamics posit just two time scales: a fast and a slow scale~\cite[e.g.]{Dolbow04, E04, Pavliotis06a}.  Indeed, Sections~\ref{sec:sacc}--\ref{sec:add} implicitly separate the dynamics of dissipative \pde{}s into the `uninteresting' fast subgrid dynamics, and the relevant slow dynamics of macroscale evolution resolved by the discretisation. But many applications possess a wide variety of interesting space-time scales~\cite[e.g.]{Brandt01, Dolbow04}.  Recent research developed a methodology with rigorous support for changing the resolved spatial grid scale by just a factor of two~\cite{Roberts08c}.  Homogenisation, in Section~\ref{sec:h}, is one example: the derivation of equation~\eqref{eq:cgddeh} recommends that the evolution of discrete diffusion on a grid, with spatially varying diffusivity, is mapped to a coarser grid as
\begin{eqnarray}&&
\frac{du_i}{dt}%=\delta\left[\kappa_i\delta u_i\right]
=\kappa_{i-1/2}u_{i-1}-(\kappa_{i-1/2}+\kappa_{i+1/2})u_i 
+\kappa_{i+1/2}u_{i+1}
\nonumber\\&\mapsto&
\frac{dU_j}{dt}\approx %\rat14\bbdelta\left[\Kappa_j\bbdelta U_j\right]=
\rat1{16}\big\{ \Kappa_{j-1}U_{j-2}-(\Kappa_{j-1}+\Kappa_{j+1})U_j +\Kappa_{j+1}U_{j+2} \big\},
\label{eq:introhomo}
\end{eqnarray}
where the coarse grid index $j=2i$ and the coarser scale diffusivity 
$\Kappa_j \approx\rat14(\kappa_{2j-2} +\kappa_{2j-1} +\kappa_{2j+1} +\kappa_{2j+2})$ is \emph{an} average over the microgrid diffusivities. The mapping of dynamics from a finer grid to a coarser grid, via finite elements formed from a small number of fine grid points, may then be iterated to generate a hierarchy of models across a wide range of spatial scales, with the theory of centre manifolds to support across the whole hierarchy.  This approach promises to empower us with great flexibility in modelling complex dynamics over multiple scales.  Sections~\ref{sec:dednc}--\ref{sec:tfair} further develop this modelling transformation of discrete dynamics on grids by exploring coupling conditions which result in the coarse grid dynamics also preserving the self-adjoint symmetries of the fine grid dynamics.

Most two scale modelling methods can also be applied over many scales.  However, most established methods require each application to be based upon a large `spectral gap': a parameter such as~$\epsilon$ measures the scale separation, and invoking ``as $\epsilon\to0$'' provides the extreme scale separation. In contrast, multigrid iteration for solving linear equations transforms between length scales that are different by (usually) a factor of two~\cite[e.g.]{Briggs01, Roberts99d, Brandt06}.  Analogously, Sections~\ref{sec:dednc}--\ref{sec:tfair} explore modelling \emph{dynamics} on a hierarchy of length scales that differ by a factor of two and hence the `spectral gap' is finite, not infinite as required by popular extant, non-multigrid, methods for modelling dynamics.  Section~\ref{sec:oepsad} describes how to divide fine grid lattice dynamics into small finite elements, develops interelement coupling rules that preserve self-adjoint symmetries, and then establishes the centre manifold support for the resulting coarse grid models.  Section~\ref{sec:tfair} outlines three applications including the homogenisation~\eqref{eq:introhomo}.

The methodology proposed here uses dynamical systems theory to support and construct accurate coarse grid models of fine scale dynamics, both continuum and discrete.  I expect that the basis developed here for deterministic dynamics in one spatial dimension can be extended to both higher dimensions and stochastic dynamics.  For non-dissipative dynamics, although the formal construction of coarse models follows analogously, current theory gives little support for the relevance of the resulting \emph{sub-centre} slow manifold models~\cite[e.g.]{Sijbrand85, Bokhove96}. This article is confined to one dimensional dissipative dynamics.

I also conjecture an implication for the equation-free modelling methodology of Kevrekidis et al.~\cite[e.g.]{Kevrekidis03b,Li03,Gear03,Samaey03b}.  When coupling sparse patches of microsimulators, previous research established that an analogous dynamical systems approach provided coupling of patches with high order accuracy on the macroscale~\cite{Roberts04d,Roberts06d}.  For dynamical systems where we want to automatically  preserve in the macroscale any microscale symmetries, the coupling condition~\eqref{eq:uxcc}, shown to be required here, suggests that each patch should have a point source at mid-patch in proportion to fluxes extrapolated from neighbouring patches.  Further research will tell.

\section{Self-adjoint preserving coupling conditions}
\label{sec:sacc}

\begin{figure}
\centering\setlength{\unitlength}{0.95ex}
\begin{picture}(77,25)
%\put(0,0){\framebox(77,25){}}
\put(1,1){% coordinate shift
\put(-1,2.5){\vector(1,0){75}\ $x$}
\multiput(4,2.5)(16,0){5}{\put(0,-0.5){\line(0,1){1}}}
\put(8,3.5){$u(x,t)$}
\put(3,0){\put(0,0){$X_{j-2}$}
  \put(16,0){$X_{j-1}$}
  \put(32,0){$X_{j}$}
  \put(48,0){$X_{j+1}$}
  \put(64,0){$X_{j+2}$}
}
\put(-1,20.5){\vector(1,0){76}}
\multiput(4,20.5)(16,0){5}{\circle*1}
\put(3,22){\put(0,0){\color{red}$U_{j-2}$}
  \put(16,0){\color{blue}$U_{j-1}$}
  \put(32,0){$U_{j}$}
  \put(48,0){\color{magenta}$U_{j+1}$}
  \put(64,0){\color{green}$U_{j+2}$}
}
\thicklines
\setcounter{rgb}{0}
\multiput(0,16)(16,-5){3}{%
  \ifcase\arabic{rgb}\color{blue}\or\or\color{magenta}\fi%
  \stepcounter{rgb}%
  \put(4,0){\line(1,0){32}}
  \put(4,0){\circle{1}}
  \put(20,0){\circle{1}}
  \put(36,0){\circle{1}}
}
\put(23,17){\color{blue}$u_{j-1}(x,t)$}
\put(39,12){$u_{j}(x,t)$}
\put(55,7){\color{magenta}$u_{j+1}(x,t)$}
}
\end{picture}
\caption{to discretise continuum dynamics, bottom, to the grid, top, rewrite the continuum dynamics of~$u(x,t)$ as the dynamics of~$u_j(x,t)$ on overlapping elements. This figure plots three consecutive elements (blue, black and magenta): the $j$th~element stretches from $X_{j-1}$~to~$X_{j+1}$.  Then analysis supports and generates rules for the evolution of grid values~$U_j(t)=u_j(X_j,t)$.}
\label{fig:sacc}
\end{figure}
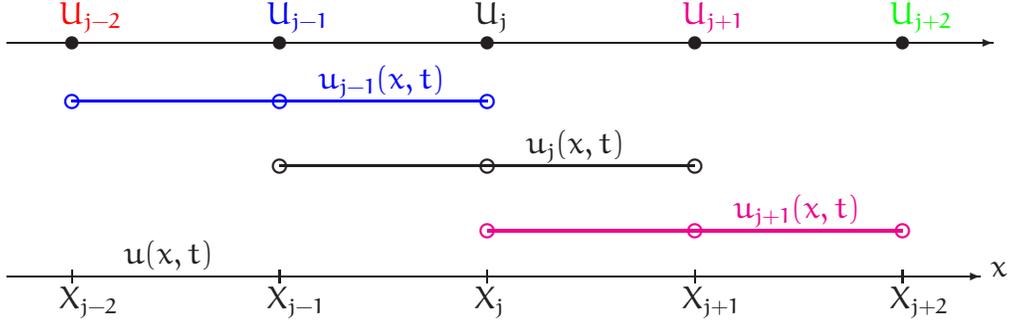

The next three sections explore the dynamics of a continuum field~$u(x,t)$ in one spatial dimension.
Figure~\ref{fig:sacc} shows a part of an $L$-periodic domain, $u(x+L,t)=u(x,t)$, which is divided into $m$~\emph{overlapping} elements $E_{j}=\{x\mid X_{j-1}\leq x\leq X_{j+1}\}$ for $m$~grid points~$X_{j}$.  These grid points need not be equally spaced.  Typically $u_{j}(x)$~and~$v_j(x)$ denote fields in the $j$th~element~$E_j$. Define the inner product
\begin{equation}
\inpr uv=\sum_{j=1}^m \int_{E_j}uv\,dx
= \sum_{j=1}^m \left\{ \int_{X_{j-1}}^{X_j^-}uv\,dx 
+\int_{X_j^+}^{X_{j+1}}uv\,dx \right\}.
\end{equation}
These integrals in the inner product are split over the two halves of each element to emphasise that perhaps unexpected contributions come from the central grid point~$X_j$.

\begin{theorem}[self-adjoint coupling] \label{thm:sac}
The operator~$\cL$ in the linear system of \emph{coupled} \pde{}s
\begin{equation}
\D t{u_j}=\cL u_j=\D x{}\left[f_j(x)\D x{u_j}\right]+\alpha g_j(x)u_j
\,,\quad x\in E_j\,,
\end{equation} 
is self-adjoint when coupled with the interelement coupling conditions
\begin{eqnarray}&&
f_j(X_j^+)u_{jx}(X_j^+)-f_j(X_j^-)u_{jx}(X_j^-)
\nonumber\\&&{}
+\gamma\big[ f_{j+1}(X_{j})u_{j+1,x}(X_j) -f_{j-1}(X_{j})u_{j-1,x}(X_j) \big]
\nonumber\\&&{}
-\gamd \big[ f_{j}(X_{j+1})u_{jx}(X_{j+1})-f_{j}(X_{j-1})u_{jx}(X_{j-1}) \big]
=0\,, \label{eq:uxcc}
\\&&
u_j(X_{j\pm1})=\gamd u_j(X_j)+\gamma u_{j\pm1}(X_{j\pm1})
\quad\text{and}
\quad u_j(X_j^-)=u_j(X_j^+)\,,
\label{eq:ucc}
\end{eqnarray}
in which subscript~$x$ denotes spatial differentiation.
\end{theorem}

Usually I link the coupling parameters by $\gamma+\gamd =1$\,; however, this constraint is not necessary in this theorem.

In many applications the coefficient functions $f_j$~and~$g_j$ do not depend upon the element~$j$.  However, in nonlinear systems we may need the reassurance of the theorem when applied with nonlinear~$f(u,u_x)$ whence $f_j=f(u_j,u_{jx})$ will be element dependent.  Any dependence upon time and other parameters are suppressed for clarity.

\begin{proof}
Undergraduate algebra proves the theorem.
Integration by parts gives
\begin{eqnarray*}
\inpr{\cL u}v
&=& \sum_j \left\{ \big[ f_ju_{jx}v_j-f_ju_jv_{jx} \big]_{X_{j-1}}^{X_j^-} 
+\big[ f_ju_{jx}v_j-f_ju_jv_{jx} \big]_{X_{j}^+}^{X_{j+1}} 
\vphantom{\int_{E_j}}\right.\\&&\quad\left.{}
+\int_{E_j}u_j\cL v_j\,dx \right\}
\\&&\text{(which upon using~\eqref{eq:ucc} becomes)}
\\&=&\inpr u{\cL v}+\sum_j\left\{ 
f_j(X_j^-)u_{jx}(X_j^-)v_j(X_j)-f_j(X_j^-)u_j(X_j)v_{jx}(X_j^-) 
\right.\\&&\left.\quad{}
-f_{j}(X_{j-1})u_{jx}(X_{j-1})\big[\gamd v_j(X_j)+\gamma v_{j-1}(X_{j-1})\big]
\right.\\&&\left.\quad{}
+f_{j}(X_{j-1})\big[\gamd u_j(X_j)+\gamma u_{j-1}(X_{j-1})\big]v_{jx}(X_{j-1})
\right.\\&&\left.\quad{}
+f_{j}(X_{j+1})u_{jx}(X_{j+1})\big[\gamd v_j(X_j)+\gamma v_{j+1}(X_{j+1})\big]
\right.\\&&\left.\quad{}
-f_{j}(X_{j+1})\big[\gamd u_j(X_j)+\gamma u_{j+1}(X_{j+1})\big]v_{jx}(X_{j+1})
\right.\\&&\left.\quad{}
-f_j(X_j^+)u_{jx}(X_j^+)v_j(X_j)+f_j(X_j^+)u_j(X_j)v_{jx}(X_j^+) 
\right\}
\\&&\text{(upon renumbering $u_{j\pm1}(X_{j\pm1})$ by $j'=j\pm 1$ becomes)}
\\&=&\inpr u{\cL v}+\sum_j\left\{
u_j(X_j)\big[ -f_j(X_j^-)v_{jx}(X_j^-) 
+\gamd f_{j}(X_{j-1})v_{jx}(X_{j-1}) 
\right.\\&&\left.\qquad{}
+\gamma f_{j+1}(X_{j})v_{j+1,x}(X_{j}) 
-\gamd f_{j}(X_{j+1})v_{jx}(X_{j+1})
\right.\\&&\left.\qquad{}
-\gamma f_{j-1}(X_j)v_{j-1,x}(X_{j}) 
+f_j(X_j^+)v_{jx}(X_j^+)
\big]
\right.\\&&\left.\quad{}
+v_j(X_j)\big[ 
f_j(X_j^-)u_{jx}(X_j^-) 
-\gamd f_{j}(X_{j-1})u_{jx}(X_{j-1}) 
\right.\\&&\left.\qquad{}
-\gamma f_{j+1}(X_{j})u_{j+1,x}(X_{j})
+\gamd f_{j}(X_{j+1})u_{jx}(X_{j+1}) 
\right.\\&&\left.\qquad{}
+\gamma f_{j-1}(X_j)u_{j-1,x}(X_j) 
-f_j(X_j^+)u_{j,x}(X_j^+)
\big]
 \right\}
 \\&=&\inpr u{\cL v} +0\quad
 \text{by \eqref{eq:uxcc}.}
\end{eqnarray*}
Hence $\cL$~with coupling conditions~\eqref{eq:uxcc}--\eqref{eq:ucc} is self-adjoint.  This proof also applies, with only minor modifications, to multi-component systems where the field~$u_j$ in each element is a vector and $f_j$~and~$g_j$ are symmetric matrices.
\end{proof}

\paragraph{Remark}
The virtue of such coupling conditions for overlapping domains is that when finding the adjoint, the algebra involves quantities that are already involved, namely the grid values.  If one tries to find the adjoint for disjoint elements, then the field values at the element boundaries also are involved, thus leading to too many free variables after the integration by parts.  The overlap of finite elements used herein are increasingly being used in multiscale modelling: examples include the `border regions' of the heterogeneous multiscale method~\cite[e.g.]{E04}, the `buffers' of the gap-tooth scheme~\cite[e.g.]{Samaey03b}, and the overlapping domain decomposition that improves convergence in waveform relaxation of parabolic \textsc{pde}s~\cite[e.g.]{Gander98}.

\section{Centre manifold theory supports discretisation}
\label{sec:cmtsd}

Based upon the equilibria and spectra about equilibria, centre manifold theory rigorously supports the existence, relevance and construction of low~dimensional models of dynamical systems~\cite[e.g.]{Carr81, Kuznetsov95}.   This section explores the theoretical support for forming discretisations of the class of self-adjoint \pde{}s in the form
\begin{equation}
\D tu=\D x{}\left[f(x,u,u_{x})\D x{u}\right] +\alpha g(x,u,u_{x})\,,
\label{eq:nde}
\end{equation} 
using the self-adjoint preserving coupling conditions~\eqref{eq:uxcc}--\eqref{eq:ucc} with $\gamd =1-\gamma$\,.

Embed the $L$-periodic dynamics of the \pde~\eqref{eq:nde} into the overlapping elements of Figure~\ref{fig:sacc} by introducing the subgrid field~$u_j(x,t)$ in each element~$E_j$.  Define these subgrid fields to satisfy the \pde~\eqref{eq:nde} in each element, namely
\begin{equation}
\D t{u_j}=\D x{}\left[f(x,u_j,u_{jx})\D x{u_j}\right] +\alpha g(x,u_j,u_{jx})\,,
\quad j=1,\ldots,m\,,
\label{eq:ndej}
\end{equation} 
and coupled by the conditions~\eqref{eq:uxcc}--\eqref{eq:ucc}. 

In the extended state space $\cE=(u_1,\ldots,u_m,\gamma,\alpha)$ there exist a subspace~$\cE_0$ of equilibria with parameters $\gamma=\alpha=0$ ($\gamd =1$) and subgrid fields constant in each element,  $u_j(x,t)=U_j$\,; that is, fields which are piecewise constant over the domain are equilibria when $\gamma=\alpha=0$\,.  For each element define $f_j(x)=f(x,U_j,0)$ then the differential equation for perturbations~$u'_j(x,t)$, linearised about the piecewise constant equilibria, is simply the diffusion equation
\begin{equation}
\D t{u'_j}=\cL u'_j =\D x{}\left[f_j(x)\D x{u'_j}\right],
\quad j=1,\ldots,m\,,
\label{eq:lde}
\end{equation}
with, as $\gamma=0$ ($\gamd =1$), the `insulating' version of the coupling conditions~\eqref{eq:uxcc}--\eqref{eq:ucc}, namely for $j=1,\ldots,m$
\begin{eqnarray}&&
f_j(X_j^+)u'_{jx}(X_j^+)-f_j(X_j^-)u'_{jx}(X_j^-)
\nonumber\\&&\quad{}
- f_{j}(X_{j+1})u'_{jx}(X_{j+1})+f_{j}(X_{j-1})u'_{jx}(X_{j-1}) 
=0\,, \label{eq:luxcc}
\\&&
u'_j(X_{j\pm1})=u'_j(X_j^-)=u'_j(X_j^+)\,.
\label{eq:lucc}
\end{eqnarray}

The centre manifold support for discrete models of the nonlinear system~\eqref{eq:nde} with coupled elements rests upon the eigenstructure of the linearised system~\eqref{eq:lde}--\eqref{eq:lucc}.  As for all self-adjoint systems, following the proof of orthogonality of eigenmodes, Section~\ref{sec:gsp}, I proceed to prove that the eigenvalues are real and they are all non-positive.  The proofs are elementary undergraduate algebra.  Then Section~\ref{sec:rsmde} uses these to prove the existence and relevance of a slow manifold, `holistic' discretisation for the nonlinear system~\eqref{eq:nde}.

\subsection{Homogeneous spectrum on an equi-spaced grid}
\label{sec:hseqg}

As one important example, this subsection considers the special case where the grid is equispaced and the `diffusivity'~$f$ only depends upon space through its dependence upon the field~$u$, that is, $f=f(u,u_x)$.  Then the linearised dynamics about the piecewise constant solutions are described by the \pde~\eqref{eq:lde} with coefficients~$f_j$ which are constant on each element.  Assume all~$f_j$ are bounded above zero.  This subsection then finds the spectrum and eigenmodes of the linearised \pde~\eqref{eq:lde}.

Seek eigenvalues~$\lambda$ in the $j$th~element such that $f_ju'_{xx}=\lambda u'$ (dropping subscript~$j$ for simplicity). As the eigenvalues must be real (see Theorem~\ref{thm:real}), set $\lambda=-f_jk^2$ for some $k\,(\geq0)$~to be determined.  For an equi-spaced grid, $\Delta X_j=h$\,, solutions must be of the form
\begin{equation}
u'=A\cos k(x-X_j)+B\sin k(x-X_j)+C\sin k|x-X_j|\,,
\label{eq:modes}
\end{equation}
upon using the continuity at $x=X_j$\,.
Then the insulating boundary conditions~\eqref{eq:lucc} of $u'(X_{j\pm1})=u'(X_j)$ require $A(\cos kh -1)+(C\pm B)\sin kh=0$\,.
The difference of these two conditions is $2B\sin kh=0$ and hence 
either $B=0$ or $kh=n\pi$\,.  Explore both in turn.
\begin{description}
\item[$kh=n\pi$] Then $\cos kh=(-1)^n$ and the two conditions reduce to simply $A\big[(-1)^n-1\big]=0$ for which two cases arise depending upon even or odd values for~$n$.
\begin{description}
\item[even $n$] Then $A$, $B$ and~$C$ are unrestrained except possibly by the derivative condition in~\eqref{eq:luxcc}.  However, all modes independently satisfy the derivative condition~\eqref{eq:luxcc}.  Hence the three orthogonal modes $u'=\cos k(x-X_j)$, $u'=\sin k(x-X_j)$ and $u'=\sin k|x-X_j|$ form a basis for the three dimensional eigenspace corresponding to eigenvalue $\lambda=-f_jk^2$\,.

\item[odd $n$] Then $A=0$\,.  Furthermore, $\sin k(x-X_j)$ satisfies the derivative condition in~\eqref{eq:luxcc}, but $\sin k|x-X_j|$ does not, $C=0$\,. Hence $u'=\sin k(x-X_j)$ is the only eigenmode.

One may like to view the modes $\sin k|x-X_j|$ for even~$n$ with wavenumber $k=n\pi/h$ as serving in place of the usual Fourier modes $\cos k(x-X_j)$ for odd~$n$.
\end{description}

\item[$kh\neq n\pi$]  In this case, necessarily $B=0$; then the condition $A(\cos kh -1)+C\sin kh=0$ and the derivative condition~\eqref{eq:luxcc} become
\begin{displaymath}
\begin{bmatrix}
\cos kh-1&\sin kh\\-\sin kh&\cos kh-1
\end{bmatrix} \begin{bmatrix}
A\\C
\end{bmatrix}
=\vec 0\,.
\end{displaymath}
The determinant 
$(\cos kh -1)^2+\sin^2kh=2(1-\cos kh)$ is non-zero except for the outlawed cases $kh=2n\pi$\,.  Hence, $A=C=0$ also, and so there are no further eigenmodes.

\end{description}

The spectrum of the linearised dynamics~\eqref{eq:lde}--\eqref{eq:lucc} (and multiplicity) on the $j$th~element is thus
\begin{equation}
\{0,-f_j\pi^2/h^2, -4f_j\pi^2/h^2 \text{ (triple)}, 
-9f_j\pi^2/h^2, -16f_j\pi^2/h^2 \text{ (triple)},\ldots\}\,.
\label{eq:simpspec}
\end{equation}
This set, with one zero eigenvalue and the rest negative, provided `diffusivity' $f_j>0$\,, are as required for the support of centre manifold theory.

\subsection{General spectral properties}
\label{sec:gsp}

This subsection proves, using elementary undergraduate algebra, that crucial properties of the spectrum~\eqref{eq:simpspec} hold for the more general system~\eqref{eq:lde}--\eqref{eq:lucc} for possibly non-uniform elements.  Alternatively, one may view the first two of these properties as a consequence of the Self-adjoint Theorem~\ref{thm:sac}, and the third as a natural consequence of the dissipation of diffusion.  A reader comfortable with such a view may proceed direct to  the next Section~\ref{sec:rsmde}.

Throughout this subsection we only address a generic $j$th~element in isolation, the interelement coupling parameter $\gamma=0$\,.  This isolation is to explore desirable properties of the dynamics in each isolated element.  The next Section~\ref{sec:rsmde} then uses centre manifold theory, based upon these results, to support the modelling of fully coupled dynamics.

\begin{theorem}[orthogonal eigenmodes]\label{thm:orthog}
Consider the operator~$\cL$ on the right-hand side of the \pde~\eqref{eq:lde} with boundary conditions~\eqref{eq:luxcc}--\eqref{eq:lucc}; eigenmodes corresponding to distinct eigenvalues are orthogonal.
\end{theorem}

\begin{proof}
Adapt the classical proof in many undergraduate texts.
Let $u$~and~$v$ be eigenmodes of~$\cL$ in the $j$th~element corresponding to eigenvalues $\lambda$ and~$\mu$, respectively.  For slight simplicity use $f$~to denote~$f_j$.  Since $\cL u=\lambda u$ and $\cL v=\mu v$\,,
\begin{eqnarray*}
0&=&
\inpr{v}{\cL u-\lambda u}
-\inpr{u}{\cL v-\mu v}
\\&&\text{(integrating by parts)}
\\&=& \big[fvu_x-fuv_x\big]_{X_j^+}^{X_{j+1}}
+\big[fvu_x-fuv_x\big]_{X_{j-1}}^{X_j^-}
+(\mu-\lambda)\inpr uv
\\&=& f(X_{j+1})v(X_{j+1})u_x(X_{j+1})
-f(X_{j+1})u(X_{j+1})v_x(X_{j+1})
\\&&{}
-f(X_j^+)v(X_j^+)u_x(X_j^+)
+f(X_j^+)u(X_j^+)v_x(X_j^+)
\\&&{}
+f(X_j^-)v(X_j^-)u_x(X_j^-)
-f(X_j^-)u(X_j^-)v_x(X_j^-)
\\&&{}
-f(X_{j-1})v(X_{j-1})u_x(X_{j-1})
+f(X_{j-1})u(X_{j-1})v_x(X_{j-1})
\\&&{}
+(\mu-\lambda)\inpr uv
\\&&\text{(by~\eqref{eq:lucc}, $u(X_j^\pm)$ is continuous and $u(X_{j\pm1})= u(X_j)$,}
\\&&\quad\text{and similarly for~$v$)}
\\&=& v(X_{j})\big[f(X_{j+1})u_x(X_{j+1})
-f(X_j^+)u_x(X_j^+)
\\&&\quad{}
+f(X_j^-)u_x(X_j^-)
-f(X_{j-1})u_x(X_{j-1}) \big]
\\&&{}
+u(X_{j})\big[ -f(X_{j+1})v_x(X_{j+1})
+f(X_j^+)v_x(X_j^+)
\\&&\quad{}
-f(X_j^-)v_x(X_j^-)
+f(X_{j-1})v_x(X_{j-1}) \big]
\\&&{}
+(\mu-\lambda)\inpr uv
\\&&\text{(by the boundary condition~\eqref{eq:luxcc} on gradients)}
\\&=&(\mu-\lambda)\inpr uv \,.
\end{eqnarray*}
Hence, for distinct eigenvalues, $\lambda\neq\mu$\,, the corresponding eigenmodes must be orthogonal, $\inpr uv =0$\,.
\end{proof}

\begin{theorem}[reality]
\label{thm:real}
Consider the operator~$\cL$ in \pde~\eqref{eq:lde} with boundary conditions~\eqref{eq:luxcc}--\eqref{eq:lucc}; its eigenvalues are all real.
\end{theorem}

\begin{proof}
Again adapt the classic proof.  Let $u$~denote any eigenmode corresponding to any eigenvalue~$\lambda$ on the $j$th~element; they are potentially complex.  Let $v$~and~$\mu$ denote the complex conjugate of $u$~and~$\lambda$, respectively.  Since the \pde~\eqref{eq:lde} and boundary conditions~\eqref{eq:luxcc}--\eqref{eq:lucc} have all real coefficients, these complex conjugates must also be an eigenmode\slash eigenvalue pair.
The derivation within the proof of Theorem~\ref{thm:orthog} then establishes that
\begin{displaymath}
(\mu-\lambda)\inpr uv=0\,.
\end{displaymath}
Here $\inpr uv=\int_{E_j} |u|^2\,dx$ which is necessarily non-zero and hence $\mu=\lambda$\,.  But $\mu$~and~$\lambda$ are complex conjugates, so any eigenvalue~$\lambda$ must be real.
\end{proof}

\begin{theorem}[spectral gap] \label{thm:gap}
Consider the operator~$\cL$ in \pde~\eqref{eq:lde} with boundary conditions~\eqref{eq:luxcc}--\eqref{eq:lucc}: when the coefficient function~$f_j(x)$ is bounded above zero, $f_{j,\min}=\min_{x\in E_j}f_j(x)>0$\,, then there is a zero eigenvalue and all other eigenvalues are negative and bounded away from zero by an amount proportional to~$f_{j,\min}/(X_{j+1}-X_{j-1})^2$.
\end{theorem}

\begin{proof}
Let $u$~denote any eigenmode corresponding to any eigenvalue~$\lambda$ on the $j$th~element.  Thus $\lambda u=[f_ju_x]_x$\,.  Multiply this \ode\ by~$u$ and integrate:
\begin{eqnarray*}
\lambda\int_{E_j}u^2dx
&=&\int_{E_j} u[f_ju_x]_x\,dx
\\&=& \big[f_juu_x\big]_{X_j^+}^{X_{j+1}}
+\big[f_juu_x\big]_{X_{j-1}}^{X_j^-}
-\int_{E_j}f_ju_x^2\,dx
\\&=& 
f_j(X_{j+1})u(X_{j+1})u_x(X_{j+1})
-f_j(X_{j}^+)u(X_{j}^+)u_x(X_{j}^+)
\\&&{}
+f_j(X_{j}^-)u(X_{j}^-)u_x(X_{j}^-)
-f_j(X_{j-1})u(X_{j-1})u_x(X_{j-1})
\\&&{}
-\int_{E_j}f_ju_x^2\,dx
\\&&\text{(by~\eqref{eq:lucc}, $u(X_j^\pm)$ is continuous and $u(X_{j\pm1})= u(X_j)$)}
\\&=& 
u(X_{j})\big[f_j(X_{j+1})u_x(X_{j+1})
-f_j(X_{j}^+)u_x(X_{j}^+)
\\&&\quad{}
+f_j(X_{j}^-)u_x(X_{j}^-)
-f_j(X_{j-1})u_x(X_{j-1}) \big]
-\int_{E_j}f_ju_x^2\,dx
\\&&\text{(by the boundary condition~\eqref{eq:luxcc} on gradients)}
\\&=&-\int_{E_j}f_ju_x^2\,dx\,.
\end{eqnarray*}
Hence, for coefficient functions~$f_j$ bounded above zero, the right-hand side is non-positive and hence so must all the eigenvalues~$\lambda$.
The eigenvalue zero corresponds only to solutions that are constant on the $j$th~element, $u_x=0$\,.

Now prove that all other eigenvalues are bounded away from zero by using a bound from the constant coefficient case of Section~\ref{sec:hseqg}.  Constrain the magnitude of the eigenmodes by $\int_{E_j}u^2dx=1$\,.  Then use the above identity to bound the eigenvalue
\begin{equation}
\lambda=-\int_{E_j}f_ju_x^2\,dx\leq -f_{j,\min} \int_{E_j}u_x^2dx\,.
\end{equation}
Thus minimise $\int_{E_j}u_x^2dx$ subject to $\int_{E_j}u^2dx=1$ and $\int_{E_j}u\,dx=0$ (as other eigenmodes are necessarily orthogonal to the constant eigenmode) and the boundary conditions~\eqref{eq:luxcc}--\eqref{eq:lucc}.  Using Lagrange multipliers $\mu$~and~$\nu$, standard Calculus of Variations asserts this minimum occurs for functions~$u(x)$ satisfying the Euler--Lagrange equation $\nu+2\mu u-2u_{xx}=0$\,.  For simplicity, let $\xi=x-X_j$ and $X_{j\pm1}-X_j=\pm\bar h+\tilde h$\,.  Three cases arise:
\begin{eqnarray*}
\mu=+k^2>0&\Rightarrow&
u=\frac{\nu}{2\mu}+A\cosh k\xi+B\sinh k\xi +C\sinh k|\xi|\,;
\\
\mu=0&\Rightarrow&
u=\rat14\nu x^2+A+B\xi +C|\xi|\,;
\\
\mu=-k^2<0&\Rightarrow&
u=\frac{\nu}{2\mu}+A\cos k\xi+B\sin k\xi +C\sin k|\xi|\,.
\end{eqnarray*}
In the first case (hyperbolic), substituting $u$~into the boundary conditions~\eqref{eq:luxcc}--\eqref{eq:lucc} and the orthogonality condition gives four linear equations for $A$, $B$, $C$ and~$\nu$ which have nontrivial values only when the determinant
\begin{displaymath}
-\frac{8\bar h}k\sinh k\bar h\,\big[ \cosh k\bar h -\cosh k\tilde h \big]=0\,.
\end{displaymath}
This has no real solutions for~$k$.  In the second case, substituting leads to the determinant $4\bar h^2(\bar h+\tilde h)(\bar h-\tilde h)$\,, which also cannot be zero for non-degenerate grids.
In the third case (trigonometric), $A$, $B$, $C$ and~$\nu$ have nontrivial values only when the determinant
\begin{displaymath}
\frac{8\bar h}k\sin k\bar h\,\big[ \cos k\bar h -\cos k\tilde h \big]=0\,.
\end{displaymath}
This only has solutions for finite~$k$---for non-degenerate grids the smallest $k=\pi/\bar h$ ---hence $\int_{E_j}u_x^2dx$ is bounded away from zero by an amount proportional to~$\bar h^{-2}\propto (X_{j+1}-X_{j-1})^{-2}$, and thus so are the negative eigenvalues.
\end{proof}

Consequently the spectrum of the linearised dynamics~\eqref{eq:lde}--\eqref{eq:lucc} on the $j$th~element is qualitatively like that of the constant coefficient case~\eqref{eq:simpspec}.  In particular, there always exists a spectral gap between the zero and the other eigenvalues.

\subsection{A relevant slow manifold discretisation exists}
\label{sec:rsmde}

Recall we aim to rigorously support discretisation of the nonlinear dynamics of the self-adjoint nonlinear general reaction diffusion equation~\eqref{eq:nde}.  For definiteness we seek spatially periodic solutions, $u(x+L,t)=u(x,t)$ for some period~$L$, and divide the domain into $m$~overlapping elements as shown schematically in Figure~\ref{fig:sacc} (with $X_m-X_0=L$).  The subgrid fields~$u_j(x,t)$ in each element satisfy the \pde~\eqref{eq:ndej} and are coupled by the conditions~\eqref{eq:uxcc}--\eqref{eq:ucc}.  Then Theorem~\ref{thm:gap} proves that linearised about any of the piecewise constant solutions of the subspace~$\cE_0$ there is a spectral gap in the dynamics of the reaction diffusion \pde~~\eqref{eq:nde}.   Consequently centre manifold theory~\cite[e.g.]{Carr81, Kuznetsov95} asserts the following.

\begin{corollary}[slow manifold] \label{thm:cmt}
For sufficiently smooth reaction~$g$ and diffusivity~$f$ in some neighbourhood of the subset of~$\cE_0$ for which $f_{j,\min}$~are bounded above zero:
\begin{enumerate}
\item there exists a $(m+2)$~dimensional slow manifold~$\cM_0$ of the subgrid \pde~\eqref{eq:ndej} coupled by~\eqref{eq:uxcc}--\eqref{eq:ucc}, one dimension for each element, and one dimension each for parameters $\gamma$~and~$\alpha$;

\item the slow manifold~$\cM_0$ may be parametrised by any reasonable measure~$U_j$ of the field in each element, that is, the slow manifold and the evolution thereon may be written, for some $u_j$~and~$g_j$, $j=1,\ldots,m$\,,  as 
\begin{equation}
u_j=u_j(\vec U,x,\alpha,\gamma)
\qtq{such that}
\dot U_j=\frac{dU_j}{dt}=g_j(\vec U,\alpha,\gamma)\,;
\label{eq:sm}
\end{equation}

\item \label{i:rel} the dynamics on~$\cM_0$ is `asymptotically complete'~\cite{Robinson96} in that for all solutions of the subgrid \pde~\eqref{eq:ndej} coupled by~\eqref{eq:uxcc}--\eqref{eq:ucc} from initial conditions in some neighbourhood of~$\cM_0$, there exists a solution of the slow manifold model~\eqref{eq:sm} that is approached exponentially quickly in time---roughly at a rate $\min_j\{f_{j,\min}/(X_{j+1}-X_{j-1})^2\}$;

\item \label{i:acc}  the order of error of an \emph{approximation} to the slow manifold~$\cM_0$ \emph{and} its evolution,~\eqref{eq:sm}, is the same as the order of the residuals of the governing \pde~\eqref{eq:ndej} and coupling~\eqref{eq:uxcc}--\eqref{eq:ucc} when evaluated at the approximation.
\end{enumerate}  
\end{corollary}

The evolution~\eqref{eq:sm} on the slow manifold~$\cM_0$ forms the discrete model of the dynamics of the \pde~\eqref{eq:ndej} coupled by~\eqref{eq:uxcc}--\eqref{eq:ucc}.   In principle, such a model is an \emph{exact closure} for the discretisation in that the model tracks the evolution from general initial conditions [Corollary~\ref{thm:cmt}.\ref{i:rel}].  However, we hardly ever can construct an exact slow manifold. Nonetheless, computer algebra readily constructs the slow manifold and its evolution to a controllable order of accuracy [Corollary ~\ref{thm:cmt}.\ref{i:acc}].  Then evaluating the model for full coupling, $\gamma=1$\,, generates a model for the dynamics of the physical reaction diffusion \pde~\eqref{eq:nde}.  Section~\ref{sec:db} provides evidence that $\gamma=1$ is within the finite domain of validity of the slow manifold~$\cM_0$.

\section{Application to discretising diffusion}
\label{sec:add}

This section briefly describes three interesting applications of the preceding Corollary~\ref{thm:cmt}.  Section~\ref{sec:ld} investigate the simplest case of linear diffusion in order to show how the slow manifold varies with coupling parameter~$\gamma$.  The interest is to see how the approach generates classic interpolation for the subgrid fields and classic finite difference rules for the evolution.  Section~\ref{sec:nd} explores nonlinear diffusion to demonstrate that this approach supports specific discretisations for nonlinear problems and that, through using the coupling conditions~\eqref{eq:uxcc}--\eqref{eq:ucc}, the discretisation preserves the self-adjointness of the original physical system.
Lastly, Section~\ref{sec:db} illustrates one way to account for domain boundary conditions other than periodic, and also verifies convergence in the coupling~$\gamma$ in one particular example.

\subsection{Linear diffusion}
\label{sec:ld}

The simplest application of the Slow Manifold Corollary~\ref{thm:cmt} is to forming discrete models of linear homogeneous diffusion as governed by the \pde
\begin{equation}
\D tu=\DD xu\,.
\label{eq:lindiff}
\end{equation}
As described in Section~\ref{sec:cmtsd}, embed the $L$-periodic spatial domain into dynamics on $m$~overlapping elements.  Corollary~\ref{thm:cmt} guarantees there exists a relevant discrete, slow manifold, model of the diffusion~\eqref{eq:lindiff}.

Computer algebra~\cite[\S2]{Roberts08j} readily constructs approximations to the slow manifold; one may check the approximation by confirming \eqref{eq:desg}--\eqref{eq:ded} satisfies the governing equations, \eqref{eq:lindiff} and \eqref{eq:uxcc}--\eqref{eq:ucc}, to the specified order of error.  Let the grid be uniform, $\Delta X_j=h$\,, and define the grid values $U_j(t)=u_j(X_j,t)$\,.
Computer algebra finds the subgrid, intraelement field is
\begin{equation}
u_j=\big[ 1+\gamma(\xi\mu\delta +\rat12|\xi|\delta^2)
+\gamma^2(\rat12\xi^2\delta^2 -\rat12|\xi|\delta^2) \big]U_j
+\Ord{\gamma^3}\,.
\label{eq:desg}
\end{equation}
in terms of the subgrid variable $\xi=(x-X_j)/h$\,, and the centred mean~$\mu$ and difference~$\delta$ operators, $\mu U_j=(U_{j+1/2}+U_{j-1/2})/2$ and $\delta U_j=U_{j+1/2}-U_{j-1/2}$\,.
Reassuringly, when evaluated at full coupling $\gamma=1$ the terms linear and quadratic in the coupling parameter~$\gamma$ form classic linear and quadratic interpolation, respectively, between the grid values~$U_j$.
The corresponding evolution on the slow manifold is the discretisation
\begin{eqnarray}
\dot U_j&=&\frac1{h^2}\gamma^2\delta^2U_j
-\frac{6-5\gamma}{12h^2}\gamma^3\delta^4U_j
+ \frac{45-75\gamma+32\gamma^2}{180h^2}\gamma^4\delta^6U_j
\nonumber\\&&{}
-\frac{210-525\gamma+448\gamma^2-130\gamma^3}{1680h^2}\gamma^5\delta^8U_j
+\Ord{\gamma^9,\delta^{10}}\,.
\label{eq:ded}
\end{eqnarray}
Evaluated at full physical coupling $\gamma=1$, this model recovers the classic centred finite difference formula for the discretisation.  Computing to higher orders in coupling~$\gamma$, gives more and more terms in the classic formula. 

This approach recovers classic formula in such simple linear dynamics.  However, the theoretical support is different: centre manifold theory applies at finite element size, to guarantee a relevant model for all initial conditions in some finite domain, and, after initial transients decay, for all times.  The only approximation is the error incurred by the truncation of the description of the slow manifold in coupling parameter~$\gamma$.

\iffalse
The equivalent differential equation is
\begin{eqnarray}&&\hspace*{-1em}
\D tU=\gamma^2\DD xU
+h^2\gamma^2(1-\gamma)\rat{1-5\gamma}{12}\Dn x4U
+h^4\gamma^2(1-\gamma) \rat{(1-2\gamma)(1-27\gamma+32\gamma^2)}{360}\Dn x6U
\nonumber\\&&\hspace*{-1em}{}
+h^6\gamma^2(1-\gamma) \rat{1 -125\gamma +1240\gamma^2 -3380\gamma^3 +3816\gamma^4 -1560\gamma^5}{20160}\Dn x8U
+\Ord{h^8,\gamma^9}\,.
\end{eqnarray}
\fi

\subsection{Nonlinear diffusion}
\label{sec:nd}

Nonlinear diffusion has many applications and is of continuing interest~\cite[e.g.]{Kurganov00, Witelski98}
As a specific example application, consider nonlinear diffusion governed by the \pde\ 
\begin{equation}
\D tu=\D x{}\left(u\D xu\right).
\label{eq:nlde}
\end{equation}
As described in Section~\ref{sec:cmtsd}, embed the $L$-periodic domain into dynamics on $m$~overlapping elements, with the coupling conditions on the gradient implemented with nonlinear diffusivity $f(x,u,u_x)=u$\,.  Corollary~\ref{thm:cmt} guarantees there exists a relevant discrete, slow manifold, model of the nonlinear diffusion~\eqref{eq:nlde}.

The rate of attraction to the slow manifold is no longer a constant: instead it is now proportional to some measure of the smallest amplitude of the initial field.  Consequently, high order asymptotic approximations to the slow manifold are replete with divisions by grid values~$U_j$.  Such divisors reflect the nonlinear diffusion and suggest that the domain of attraction of the slow manifold model reduces for fields~$u$ of small magnitude.

Computer algebra~\cite[\S3]{Roberts08j} readily constructs the subgrid field of the slow manifold:
\begin{eqnarray}
u_j&=&\big[ 1+\gamma\xi\mu\delta +\rat12\gamma(1-\gamma)|\xi|\delta^2)
+\gamma^2\rat12\xi^2\delta^2 \big]U_j
\nonumber\\&&{}
+\gamma^2\big[ \xi(1-|\xi|)(\mu\delta U_j +\rat1{2U_j}\mu\delta U_j^2)
+(|\xi|-\xi^2)(\rat12\delta^2U_j-\rat1{4U_j}\delta^2 U_j^2) \big]
\nonumber\\&&{}
+\Ord{\gamma^3}\,.
\end{eqnarray}
The first line is identical to that for linear diffusion,~\eqref{eq:desg}; thus the second line is due to the nonlinearity in the diffusion.  Unlike usual finite differences or finite elements methodology which imposes an interpolation between grid values, part of the value of this `holistic discretisation'~\cite[e.g.]{Roberts00a} is that the subgrid field is constructed to satisfy the \pde~\eqref{eq:nlde} and thus generates accurate closures for the discrete model.  The evolution on the slow manifold gives the discrete model
\begin{equation}
\dot U_j=\frac1{2h^2}\gamma^2\delta^2U_j^2
-\frac{6-5\gamma}{24h^2}\gamma^3\delta^4U_j^2
+ \frac{45}{360h^2}\gamma^4\delta^6U_j^2
+\Ord{\gamma^5}\,.
\label{eq:nded}
\end{equation}
Truncated to errors~$\Ord{\gamma^3}$ then evaluated at full coupling, $\gamma=1$\,, this is the introductory model~\eqref{eq:introdiff}.
To errors of order~$\Ord{\gamma^5}$, the discretisation~\eqref{eq:nded} of the nonlinear diffusion is the same as the discretisation~\eqref{eq:ded} of linear diffusion, but applied to~$\rat12U_j^2$ instead of to~$U_j$.  This is reasonable since this continuum nonlinear diffusion operator~$(uu_x)_x=(\rat12u^2)_{xx}$.  This nontrivial correspondence, not imposed at all but instead a natural closure from the intricate subgrid scale interactions, confirms this approach to discretisation is sound.

However, such a very close correspondence breaks down at~$\Ord{\gamma^5}$ when divisions by~$U_{j\pm1}$ start invading the slow evolution~\eqref{eq:nded} of the nonlinear diffusion.  Subtleties in such higher order subgrid scale interactions result in such more complicated discretisations.

Nonetheless, the equivalent differential equation of~\eqref{eq:nded} is
\begin{eqnarray}
\D tU&=&\gamma^2\D x{}\left(U\D xU\right)
+h^2\gamma^2(1-\gamma)\frac{1-5\gamma}{12}\D x{}\left(3\D xU\DD xU+U\Dn x3U\right)
\nonumber\\&&{}
+\Ord{h^4,\gamma^5}\,,
\label{ndeede}
\end{eqnarray}
which is in conservative form no matter what order we truncate the analysis in coupling parameter~$\gamma$. The coupling conditions~\eqref{eq:uxcc}--\eqref{eq:ucc} do preserve conservation.

\subsection{Dirichlet boundaries on one element}
\label{sec:db}

In every other section we explore dynamics far away from physical boundaries by assuming spatial periodicity.  Conversely, this section explores the extreme case of precisely one element between physical boundaries forming that one element, say the element is non-dimenisonalised to $-1<x<1$\,.
This extreme case empowers us to compute to high order and show convergence in the `coupling' parameter~$\gamma$, as well as illustrating the ease of incorporating physical boundary conditions on the global domain rather than assuming periodic conditions for the $m$~elements as used elsewhere.

For definiteness, suppose the physical boundary conditions for the example \pde\ of nonlinear diffusion~\eqref{eq:nlde} are the Dirichlet conditions that $u=0$ at $x=\pm1$\,. Implement these Dirichlet conditions on the one element~$[-1,1]$ via the adaptation of the coupling conditions~\eqref{eq:uxcc}--\eqref{eq:ucc} to
\begin{eqnarray}&&
u(0^+,t)u_{x}(0^+,t)-u(0^-,t)u_{x}(0^-,t)
\nonumber\\&&{}
=\gamd \big[ u(1,t)u_{x}(1,t)-u(-1,t)u_{x}(-1,t) \big], 
\label{eq:uxcc1}
\\&&
u(\pm1,t)=\gamd u(0,t)
\quad\text{and}
\quad u(0^-,t)=u(0^+,t)\,.
\label{eq:ucc1}
\end{eqnarray}
When coupling parameter~$\gamma=0$ ($\gamd =1$) the element is isolated from the physical boundaries and the spectrum of the dynamics on the one element are as for the previous Section~\ref{sec:nd}.  Thus there exists a slow manifold parametrised by~$\gamma$ and the mid-element value $U_0(t)=u(0,t)$\,:  the slow manifold may be described by $u(x,t)=u_0(x,U_0,\gamma)$ such that $\dot U_0=g(U_0,\gamma)$\,.  It is exponentially quickly attractive, roughly at a rate~${}\propto U_0$\,.

Straightforward adaptions of the computer algebra~\cite{Roberts08j} construct the slow manifold and the evolution thereon to high order in the coupling parameter~$\gamma$.
The slow manifold is 
\begin{equation}
u_0=U_0\left[1 -\gamma|x| +\gamma^2(|x|-x^2) +\gamma^3\rat1{12}(|x|+9x^2-10|x|^3)
+\Ord{\gamma^4}\right].
\end{equation}
Observe the $\gamma^1$~and~$\gamma^2$ subgrid structures give classic linear and quadratic interpolation when evaluated at the physical $\gamma=1$\,.  It is the $\Ord{\gamma^3}$~terms that begins to account for the nonlinear nature of the diffusion, hence provide correct subgrid structures, and consequently provide a sound closure for the macroscale, one dimensional model on the element.
The corresponding evolution on the slow manifold is 
\begin{equation}
\dot U_0=-U_0^2\left[ \gamma^2+\rat12\gamma^3 -\rat14\gamma^5 -\rat{11}{48}\gamma^6 -\rat{7}{96}\gamma^7 +\rat{55}{864}\gamma^8 +\rat{89}{864}\gamma^9 +\Ord{\gamma^{10}} \right].
\end{equation}
This predicts that from all nearby initial conditions, and apart from exponentially quick transients, solutions decay algebraically in proportion to~$1/t$ as `material'~$u$ diffuses through the physical boundaries at $x=\pm 1$ where the diffusivity is zero but the flux~$u\D xu$ is not.

\begin{figure}
\centering
\begin{tabular}{c@{}c}
\rotatebox{90}{\hspace*{15ex}$1/r^2$}&
\includegraphics{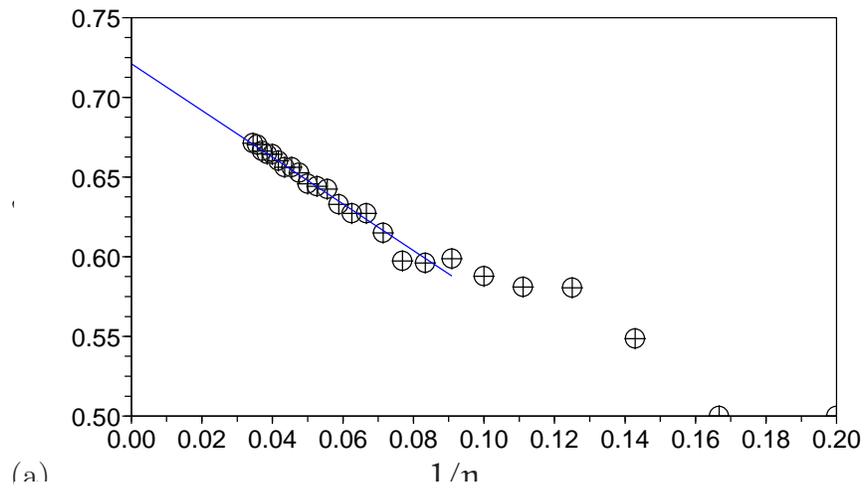}\\[-1ex]
(a)& $1/n$\\[2ex]
\rotatebox{90}{\hspace*{15ex}$\cos\text{angle}$}&
\includegraphics{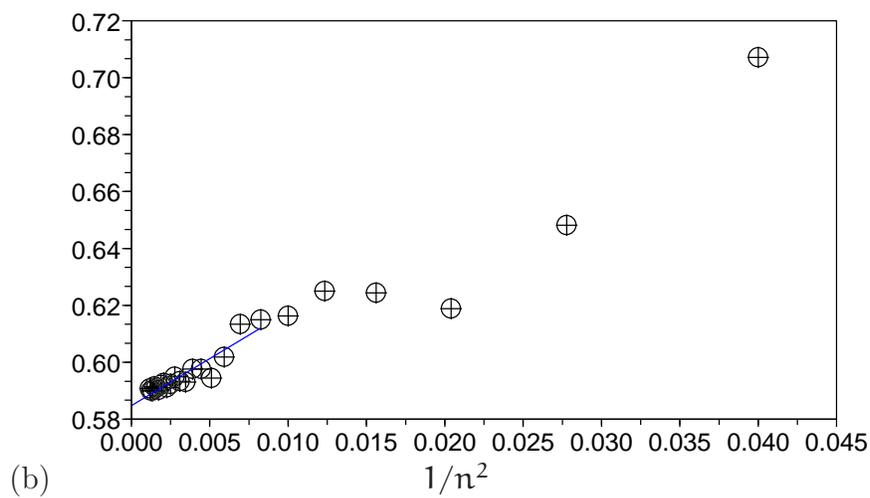}\\[-1ex]
(b)& $1/n^2$
\end{tabular}
\caption{power series to errors~$\Ord{\gamma^{30}}$ generate these generalised Domb--Sykes plots~\cite[Appendix]{Mercer90} that, when extrapolated to $1/n=0$\,, strongly suggest the power series in coupling parameter~$\gamma$ has radius of convergence~$1.18$.}
\label{fig:dsr}
\end{figure}

With just one `amplitude'~$U_0$ computer algebra generates approximations to~$\Ord{\gamma^{30}}$ within a minute.  Then generalised Domb--Sykes plots~\cite[Appendix]{Mercer90} estimate the location of the convergence limiting singularity.
Figure~\ref{fig:dsr} strongly suggests convergence in coupling parameter~$\gamma$ with radius of convergence~$1.18$ due to a complex conjugate pair of (logarithmic) singularities in complex~$\gamma$ at an angle of~$54^\circ$ to the real $\gamma$~axis.  These plots provide good evidence that evaluation at $\gamma=1$ of the power series' is convergent for at least this simple case.

\section{Disjoint lattice elements do not couple}
\label{sec:dednc}

Research into spatio-temporal dynamics commonly invokes a discrete lattice~\cite[e.g.]{Chen98, Cisternas03, Mobilia06, Simpson07, Giannoulis08}.
Section~\ref{sec:oepsad} introduces rigorous support for a transformation from fine scale lattice dynamics to a coarser scale lattice dynamics.  Similar to the continuum case of Sections~\ref{sec:sacc}--\ref{sec:add}, the transformation requires embedding the dynamics in a system of twice the dimensionality, before reducing the dimension by a factor of four, to achieve an overall halving of the dimensionality. The resultant model then resolves dynamics on a lattice with twice the grid spacing.  This section argues that it is difficult, if not impossible, to achieve such net halving of the dimensionality without the initial doubling via the embedding. 

As the most basic, but key, example of self-adjoint dynamics on a lattice, consider the discrete diffusion equation nondimensionalised to
\begin{equation}
\dot u_i=u_{i-1}-2u_i+u_{i+1}\,,
\label{eq:dde}
\end{equation}
on equi-spaced grid points $x_i=ih$\,.
This section seeks to construct a sound coarse scale model for these dynamics, but fails because of interesting reasons that inspire the next sections.

For almost extreme simplicity, suppose the discrete diffusion~\eqref{eq:dde} applies on a small finite domain at just four lattice points, $i=1,2,3,4$, with insulating Neumann-like boundary conditions provided by
\begin{equation}
u_1-u_0=u_4-u_5=0\,.
\label{eq:nlbc}
\end{equation}
We seek a model in just two dynamical variables, for this system with four dynamical variables.  Because the dynamics are so low dimensional, and linear, we reasonably explore all options.

\begin{figure}
\centering\setlength{\unitlength}{1ex}
\begin{picture}(40,10)
%\put(0,0){\framebox(40,10){}}
\put(0,7){\line(1,0){24}}
\put(0,7){\circle{1}}
\put(8,7){\circle*{1}}
\put(16,7){\circle*{1}}
\put(24,7){\circle{1}}
\put(16,3){\line(1,0){24}}
\put(16,3){\circle{1}}
\put(24,3){\circle*{1}}
\put(32,3){\circle*{1}}
\put(40,3){\circle{1}}
\put(-1,0.5){\put(0,0){$u_0$}
\put(8,0){$u_1$}
\put(16,0){$u_2$}
\put(24,0){$u_3$}
\put(32,0){$u_4$}
\put(40,0){$u_5$}
}
\put(-1,8.5){\put(0,0){$v_0$}
\put(8,0){$v_1$}
\put(16,0){$v_2$}
\put(24,0){$v_3$}
}
\put(15,4.5){\put(0,0){$v_4$}
\put(8,0){$v_5$}
\put(16,0){$v_6$}
\put(24,0){$v_7$}
}
\end{picture}
\caption{rewrite the dynamics of~$\vec u$ as the dynamics of~$\vec v$ on two elements by renaming variables: the solid discs correspond to differential equations, and the open circles correspond to algebraic coupling equations.}
\label{fig:dednc}
\end{figure}
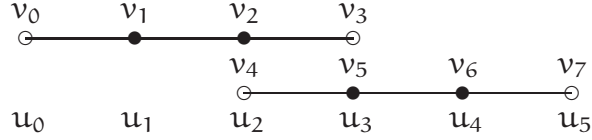

Divide the domain into two elements, the first containing $u_1$~and~$u_2$ and the second containing $u_3$~and~$u_4$. As shown in Figure~\ref{fig:dednc}, to connect the two elements introduce new names $v_3=v_5=u_3$ and $v_2=v_4=u_2$ for the middle two dynamical variables.  For convenience, rename the others variables $v_0=u_0$\,, $v_1=u_1$\,, $v_6=u_4$ and $v_7=u_5$\,.  Then write the discrete diffusion~\eqref{eq:dde} and insulating boundary conditions~\eqref{eq:nlbc}, together with the interelement coupling identities, as the differential-algebraic system
\begin{equation}
D\dot{\vec v}=L\vec v\,,
\label{eq:dde2}
\end{equation}
where $\vec v=(v_0,\ldots,v_7)$, $D=\operatorname{diag}(0,1,1,0,0,1,1,0)$ and the matrix
\begin{equation}
\newcommand{\B}{\color{blue}\it}
L=\begin{bmatrix} \B-1&1&0&\B0&\B0&0&0&\B0\\
       1&-2&1&0&0&0&0&0\\
       0&1&-2&1&0&0&0&0\\
       \B0&0&1&\B0&\B-1&0&0&\B0\\
       \B0&0&0&\B-1&\B0&1&0&\B0\\
       0&0&0&0&1&-2&1&0\\
       0&0&0&0&0&1&-2&1\\
       \B0&0&0&\B0&\B0&0&1&\B-1
\end{bmatrix},
\label{eq:ll2}
\end{equation}
where I explain the import of the italic entries shortly.
The spectrum of the linear discrete diffusion~\eqref{eq:dde2}, from $\det(L-\lambda D)$, is $\{0,-2+\sqrt2,-2,-2-\sqrt2\}$.  We construct a long term model from the two slow modes corresponding to the two eigenvalues nearest zero.  Thus we seek a model with just two dynamical variables that systematically track the amplitude of the two slowest modes.

Centre manifold theory provides rigorous support for such low dimensional modelling~\cite[e.g.]{Carr81, Kuznetsov95}.  But, analogously to the continuum analysis of Sections~\ref{sec:sacc}--\ref{sec:add}, we need to artificially modify the linear operator~$L$ to have two eigenvalues of zero, then implement a homotopy (a smooth path) to recover~$L$ and the original dynamics~\cite[e.g.]{Roberts00a, Roberts08c}.  I argue this is impossible, and hence we need the more complicated embedding of the next section.

To see what freedom we have available, first identify the aspects of~$L$ that cannot be changed.  We posit that the evolution equations, corresponding to the second, third, sixth and seventh lines of~$L$ cannot be changed as they are to encode the microscale dynamics~\eqref{eq:dde}: if we need to modify the microscale dynamics, then any embeddings we find are almost certainly problem specific and thus not of general power.   To maintain the self-adjoint symmetry of~$L$, this then fixes $L_{1,2}$, $L_{4,3}$, $L_{5,6}$ and~$L_{8,7}$ to be one, and many other entries to be zero.  The entries to vary are those in italics in~\eqref{eq:ll2}.  Second, of these, set $L_{1,8}=L_{8,1}=0$ to avoid excessively nonlocal equations.  Third, preserving the zero eigenvalue of $u_i=\text{constant}$, the conservation mode, we need each row sum to remain zero.  Lastly, as well as self-adjoint symmetry, we require isotropy, left-right symmetry. These four requirements result in three degrees of freedom spanned by the three matrices
\begin{eqnarray*}&&
Y_1=\begin{bmatrix}
0 & 0 & 0 & 0 & 0 & 0 & 0 & 0\\
0 & 0 & 0 & 0 & 0 & 0 & 0 & 0\\
0 & 0 & 0 & 0 & 0 & 0 & 0 & 0\\
0 & 0 & 0 & -1 & 1 & 0 & 0 & 0\\
0 & 0 & 0 & 1 & -1 & 0 & 0 & 0\\
0 & 0 & 0 & 0 & 0 & 0 & 0 & 0\\
0 & 0 & 0 & 0 & 0 & 0 & 0 & 0\\
0 & 0 & 0 & 0 & 0 & 0 & 0 & 0
\end{bmatrix},
\\&&
Y_2=\begin{bmatrix}
1 & 0 & 0 & -1 & 0 & 0 & 0 & 0 \\
0 & 0 & 0 & 0 & 0 & 0 & 0 & 0 \\
0 & 0 & 0 & 0 & 0 & 0 & 0 & 0 \\
-1 & 0 & 0 & 1 & 0 & 0 & 0 & 0 \\
0 & 0 & 0 & 0 & 1 & 0 & 0 & -1\\
0 & 0 & 0 & 0 & 0 & 0 & 0 & 0 \\
0 & 0 & 0 & 0 & 0 & 0 & 0 & 0 \\
0 & 0 & 0 & 0 & -1 & 0 & 0 & 1 
\end{bmatrix},
\\&&
Y_3=\begin{bmatrix}
0 & 0 & 0 & -1 & 1 & 0 & 0 & 0 \\
0 & 0 & 0 & 0 & 0 & 0 & 0 & 0 \\
0 & 0 & 0 & 0 & 0 & 0 & 0 & 0 \\
-1 & 0 & 0 & 0 & 0 & 0 & 0 & 1 \\
1 & 0 & 0 & 0 & 0 & 0 & 0 & -1\\
0 & 0 & 0 & 0 & 0 & 0 & 0 & 0 \\
0 & 0 & 0 & 0 & 0 & 0 & 0 & 0 \\
0 & 0 & 0 & 1 & -1 & 0 & 0 & 0 
\end{bmatrix}.
\end{eqnarray*}
Thus the (nearly) most general self-adjoint, isotropic, conservative, matrix of the discrete diffusion dynamics is $L+y_1Y_1+y_2Y_2+y_3Y_3$ for some constants~$y_1$, $y_2$ and~$y_3$.  Its spectrum is given by the zeros of the characteristic polynomial
\begin{equation}
\det(L+y_1Y_1+y_2Y_2+y_3Y_3-\lambda D)
=c_1\lambda+c_2\lambda^2+c_3\lambda^3+c_4\lambda^4\,,
\label{eq:chpoly2}
\end{equation}
for some coefficients~$c_k$ depending upon~$y_1$, $y_2$ and~$y_3$: for example,
\begin{displaymath}
c_4=-(1-2y_2)(1+4y_3^2+4y_2y_3-2y_1+2y_1y_2)\,.
\end{displaymath}
\begin{figure}
\centering
\includegraphics{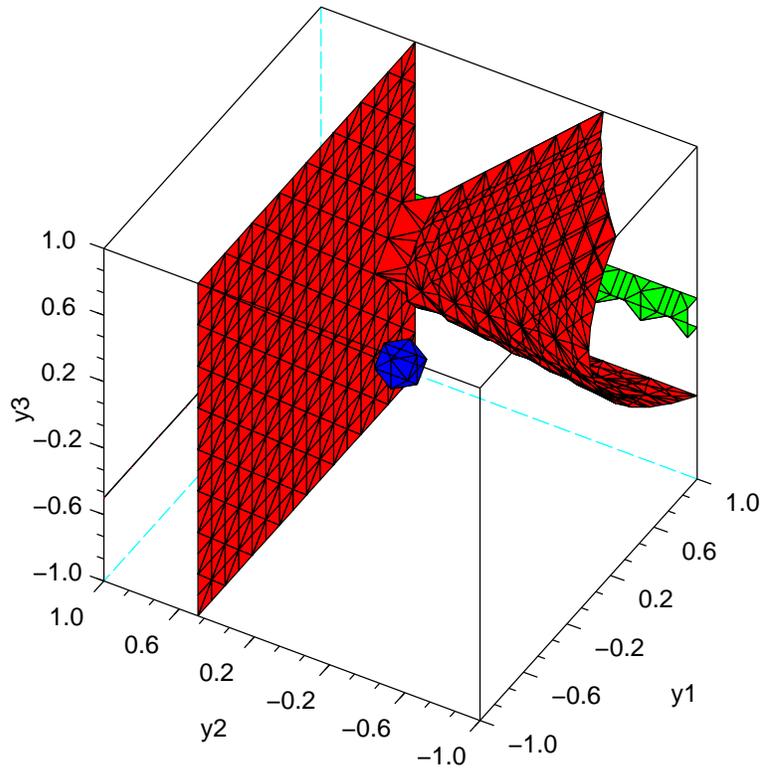}
\caption{red: iso-surface of `infinite' eigenvalues from the coefficient $c_4=0$ of the characteristic polynomial~\eqref{eq:chpoly2}.  Blue blob: the origin indicating the original diffusion dynamics~\eqref{eq:ll2}.  Green line: two decoupled elements.  There appears no route from green to blue without encountering `infinite' eigenvalues, red.}
\label{fig:dbl4eigs}
\end{figure}%
Recall that we seek to create a homotopy from dynamics with a double zero eigenvalue on two decoupled elements, $y_1=1$ and $y_3=0$ (the `green line' in Figure~\ref{fig:dbl4eigs}), to the original dynamics with $y_1=y_2=y_3=0$ (the `blue blob' in Figure~\ref{fig:dbl4eigs}).  Figure~\ref{fig:dbl4eigs} indicates that it is impossible to create a homotopy from the two decoupled elements to the original dynamics without crossing the red surfaces.  The red surfaces in Figure~\ref{fig:dbl4eigs} are the surfaces of $c_4=0$ in the charateristic polynomial~\eqref{eq:chpoly2}, and hence represent neighbourhoods where eigenvalues become infinitely large, both positive and negative.  We cannot create a smooth homotopy from two decoupled elements to the coupled original dynamics through such neighbourhoods.  This failure with just two elements suggests that, within the class of self-adjoint, isotropic, diffusion dynamics, we cannot artificially divide a domain into disjoint elements.

Thus the next section proceeds to explore embedding lattice dynamics onto overlapping elements analogous to the overlapping elements used for the continuum dynamics of Sections~\ref{sec:sacc}--\ref{sec:add} and for other multiscale approaches~\cite[e.g.]{E04, Samaey03b, Gander98}.

\section{Overlapping elements preserve self-adjoint dynamics}
\label{sec:oepsad}

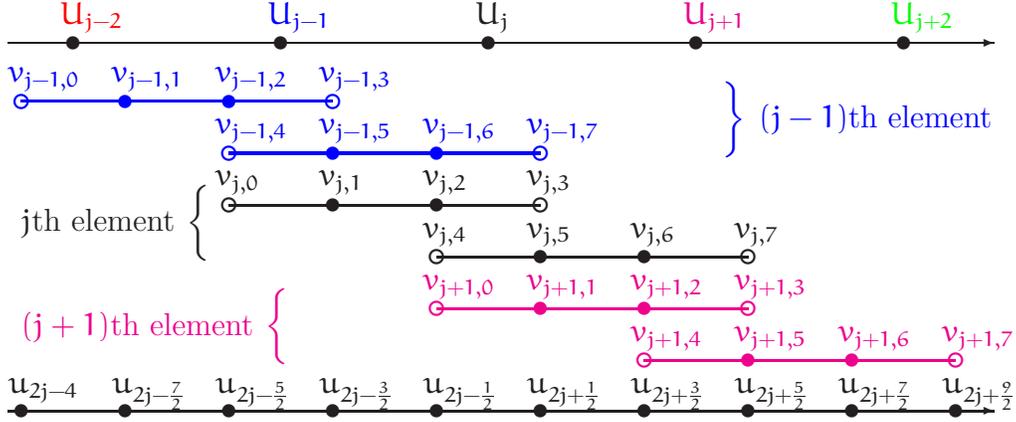
\begin{figure}
\centering\setlength{\unitlength}{0.95ex}
\begin{picture}(77,32)
%\put(0,0){\framebox(77,32){}}
\put(1,1){% coordinate shift
\put(-1,-0.9){\vector(1,0){76}}
\multiput(0,-0.9)(8,0){10}{\circle*{1}}
\put(-1,0.5){\put(0,0){$u_{2j-4}$}
  \put(8,0){$u_{2j-\frac72}$}
  \put(16,0){$u_{2j-\frac52}$}
  \put(24,0){$u_{2j-\frac32}$}
  \put(32,0){$u_{2j-\frac12}$}
  \put(40,0){$u_{2j+\frac12}$}
  \put(48,0){$u_{2j+\frac32}$}
  \put(56,0){$u_{2j+\frac52}$}
  \put(64,0){$u_{2j+\frac72}$}
  \put(72,0){$u_{2j+\frac92}$}
}
\put(-1,27.5){\vector(1,0){76}}
\multiput(4,27.5)(16,0){5}{\circle*{1}}
\put(3,29){\put(0,0){\color{red}$U_{j-2}$}
  \put(16,0){\color{blue}$U_{j-1}$}
  \put(32,0){$U_{j}$}
  \put(48,0){\color{magenta}$U_{j+1}$}
  \put(64,0){\color{green}$U_{j+2}$}
}
\thicklines
\setcounter{rgb}{0}
\multiput(0,16)(16,-8){3}{%
  \ifcase\arabic{rgb}\color{blue}\or\or\color{magenta}\fi%
  \stepcounter{rgb}%
  \put(0,7){\line(1,0){24}}
  \put(0,7){\circle{1}}
  \put(8,7){\circle*{1}}
  \put(16,7){\circle*{1}}
  \put(24,7){\circle{1}}
  \put(16,3){\line(1,0){24}}
  \put(16,3){\circle{1}}
  \put(24,3){\circle*{1}}
  \put(32,3){\circle*{1}}
  \put(40,3){\circle{1}}
}
\put(0,13){$j$th~element $\left\{\vphantom{\displaystyle\int}\right.$}
\put(0,5){\color{magenta}
  $(j+1)$th~element $\left\{\vphantom{\displaystyle\int}\right.$}
\put(54,21){\color{blue} $\left\}\vphantom{\displaystyle\int}\right.$
  $(j-1)$th~element}
\put(16,8){
  \put(-1,8.5){\put(0,0){$v_{j,0}$}
    \put(8,0){$v_{j,1}$}
    \put(16,0){$v_{j,2}$}
    \put(24,0){$v_{j,3}$}
  }
  \put(15,4.5){\put(0,0){$v_{j,4}$}
    \put(8,0){$v_{j,5}$}
    \put(16,0){$v_{j,6}$}
    \put(24,0){$v_{j,7}$}
  }
}
\put(0,16){\color{blue}
  \put(-1,8.5){\put(0,0){$v_{j-1,0}$}
    \put(8,0){$v_{j-1,1}$}
    \put(16,0){$v_{j-1,2}$}
    \put(24,0){$v_{j-1,3}$}
  }
  \put(15,4.5){\put(0,0){$v_{j-1,4}$}
    \put(8,0){$v_{j-1,5}$}
    \put(16,0){$v_{j-1,6}$}
    \put(24,0){$v_{j-1,7}$}
  }
}
\put(32,0){\color{magenta}
  \put(-1,8.5){\put(0,0){$v_{j+1,0}$}
    \put(8,0){$v_{j+1,1}$}
    \put(16,0){$v_{j+1,2}$}
    \put(24,0){$v_{j+1,3}$}
  }
  \put(15,4.5){\put(0,0){$v_{j+1,4}$}
    \put(8,0){$v_{j+1,5}$}
    \put(16,0){$v_{j+1,6}$}
    \put(24,0){$v_{j+1,7}$}
  }
}
}
\end{picture}
\caption{to transform dynamics from the fine grid, bottom, to the coarse grid, top, rewrite the dynamics of~$u_i$ as the dynamics of~$v_{j,i}$ on overlapping elements, here see three consecutive elements (blue, black and magenta), by duplicating and renaming variables: the solid discs correspond to differential equations showing the number of dynamic variables is doubled, and the circles correspond to algebraic coupling equations.  Each element has two halves, also shown separated for clarity.}
\label{fig:oepsad}
\end{figure}

This section explores how to transform the discrete dynamics of variables~$u_i(t)$ on a fine grid of spacing~$h$, into discrete dynamics of variables~$U_j(t)$ on a coarser grid of spacing~$H=2h$.

Figure~\ref{fig:oepsad} schematically shows that the $j$th~element stretches from a neighbourhood of~$X_{j-1}$ to a neighbourhood of~$X_{j+1}$.  Figure~\ref{fig:oepsad} also shows each element is divided into two halves, and the variables duplicated so that, notionally, $u_{2j}=v_{j-1,6}=v_{j,2}=v_{j,4}=v_{j+1,0}$ and $u_{2j+1}=v_{j-1,7}=v_{j,3}=v_{j,5}=v_{j+1,1}$\,.  Embed the fine grid dynamics in these overlapping elements in a space of double the dimensionality by treating $v_{j-1,6}$~and~$v_{j,2}$, and  $v_{j,5}$~and~$v_{j+1,1}$ as independently evolving variables.

In the embedding space we use the inner product
$\inpr{\vec v}{\vec w}=\sum_{j,i} v_{j,i}w_{j,i}$\,.

\subsection{Full coupling rules}

Analogous to the continuum dynamics of Sections~\ref{sec:sacc}--\ref{sec:add}, we need to find rules to couple neighbouring overlapping elements that are sufficiently generic that they are useful for modelling a wide variety of self-adjoint lattice dynamics.   Here we focus on the basic case of the coarse grid modelling of the discrete diffusion dynamics~\eqref{eq:dde}: equation~\eqref{eq:dde} applies at the internal dynamic variables $v_{j,1}$, $v_{j,2}$, $v_{j,5}$ and~$v_{j,6}$ of each element, the discs in Figure~\ref{fig:oepsad}.  Now take up the challenge of using the variables  $v_{j,0}$, $v_{j,4}$, $v_{j,5}$ and~$v_{j,7}$, the open circles in Figure~\ref{fig:oepsad}, to couple together not only the elements but also the two halves of each element.

Again analogous to the continuum dynamics of Sections~\ref{sec:sacc}--\ref{sec:add},  where rigorous theorems are invoked we require that the lattice dynamics are periodic in the fine grid with period~$2m$ in~$i$, and thus periodic on the coarse grid with period~$m$ in~$j$.  

Let us identify the possible domain of evolution and coupling rules.
Following Section~\ref{sec:dednc}, one evolution rule within each element is
\begin{equation}
D\dot{\vec v}_j=L\vec v_j
\label{eq:ddej}
\end{equation}
where $\vec v_j=(v_{j,0},\ldots,v_{j,7})$, $D=\operatorname{diag}(0,1,1,0,0,1,1,0)$ and equation~\eqref{eq:ll2} gives the matrix~$L$.  As in Section~\ref{sec:dednc}, to preserve the self-adjoint isotropic conservation of the dynamics the upright entries in the matrix~$L$ are fixed; we can only modify the italic entries. As in Section~\ref{sec:dednc}, one set of possible changes to~$L$ is spanned by $Y_1$, $Y_2$ and~$Y_3$.  This set changes only internal interactions.  There are also three basic matrices that couple elements together \emph{preserving self-adjoint isotropic} conservation, with the additional constraint that the coupling has to be \emph{`local' on the grid}.  Let $\bshft\pm$~denote the shift operators from one element to its neighbours: define $\bshft\pm v_{j,i}=v_{j\pm1,i}$\,.  Then write three basis matrices for nearest neighbour coupling changes to~$L$ as
\begin{eqnarray*}&&
Y_4=\begin{bmatrix}
1 & 0 & 0 & 0 & 0 & 0 & 0 & -\bshft-\\
0 & 0 & 0 & 0 & 0 & 0 & 0 & 0\\
0 & 0 & 0 & 0 & 0 & 0 & 0 & 0\\
0 & 0 & 0 & 0 & 0 & 0 & 0 & 0\\
0 & 0 & 0 & 0 & 0 & 0 & 0 & 0\\
0 & 0 & 0 & 0 & 0 & 0 & 0 & 0\\
0 & 0 & 0 & 0 & 0 & 0 & 0 & 0\\
-\bshft+ & 0 & 0 & 0 & 0 & 0 & 0 & 1
\end{bmatrix},
\\&&
Y_5=\begin{bmatrix}
-1 & 0 & 0 & 0 & \bshft- & 0 & 0 & 0 \\
0 & 0 & 0 & 0 & 0 & 0 & 0 & 0 \\
0 & 0 & 0 & 0 & 0 & 0 & 0 & 0 \\
0 & 0 & 0 & 0 & 0 & 0 & 0 & \bshft- \\
\bshft+ & 0 & 0 & 0 & 0 & 0 & 0 & 0\\
0 & 0 & 0 & 0 & 0 & 0 & 0 & 0 \\
0 & 0 & 0 & 0 & 0 & 0 & 0 & 0 \\
0 & 0 & 0 & \bshft+ & 0 & 0 & 0 & -1 
\end{bmatrix},
\\&&
Y_6=\begin{bmatrix}
1 & 0 & 0 & -\bshft- & 0 & 0 & 0 & 0 \\
0 & 0 & 0 & 0 & 0 & 0 & 0 & 0 \\
0 & 0 & 0 & 0 & 0 & 0 & 0 & 0 \\
-\bshft+ & 0 & 0 & 0 & 0 & 0 & 0 & 0 \\
0 & 0 & 0 & 0 & 0 & 0 & 0 & -\bshft-\\
0 & 0 & 0 & 0 & 0 & 0 & 0 & 0 \\
0 & 0 & 0 & 0 & 0 & 0 & 0 & 0 \\
0 & 0 & 0 & 0 & -\bshft+ & 0 & 0 & 1 
\end{bmatrix}.
\end{eqnarray*}
For example, the non-zero elements of $Y_4\vec v_j$ are $v_{j,0}-v_{j-1,7}$ and $v_{j,7}-v_{j+1,0}$ which connects end values of neighbouring overlapping elements.  Thus we explore the self-adjoint isotropic conservative dynamics of
\begin{equation}
D\dot{\vec v}_j=\left(L+\sum_{l=1}^6 y_lY_l\right)\vec v_j\,,
\label{eq:ddejx}
\end{equation}
for some coefficients~$y_l$ that we are free to choose.

When fully coupled, the element dynamics~\eqref{eq:ddejx} must reduce to that of the discrete diffusion~\eqref{eq:dde}.  These are most compactly expressed in terms of the fine grid shift operator~$\shft\pm$ defined as $\shft\pm v_{j,i}=v_{j,i\pm1}$\,.  (When used appropriately, the coarse grid shift $\bshft\pm=\shft\pm^2$\,.)  Then the fine grid diffusion, expressed as $\dot u_i=(\shft+-2+\shft-)u_i$\,, has `eigenvalue'~$(\shft+-2+\shft-)$ corresponding to `eigenvector' $\vec u=(\ldots,\shft-^2,\shft-,1,\shft+,\shft+^2,\ldots)$ where all are interpreted in an operator sense.  For the element dynamics~\eqref{eq:ddejx} to reproduce these dynamics exactly, the dynamics must have the same `eigenvalue'~$(\shft+-2+\shft-)$ but corresponding to the `eigenvector' $\vec v_j=\bshft+^j(1,\shft+,\shft+^2,\shft+^3,\shft+^2,\shft+^3,\shft+^4,\shft+^5)$.  Substituting these into~\eqref{eq:ddejx}, equating coefficients of $\shft\pm$ in all components gives a set of equations that uniquely determine $y_1=y_5=y_6=1$ and $y_2=y_3=y_4=0$\,.  Thus the unique operator that gives the correct evolution of the diffusion~\eqref{eq:dde} with the correct subgrid microstructure for the diffusion is
\begin{equation}
L_1=\begin{bmatrix} -1&1&0&-\bshft-&\bshft-&0&0&0\\
       1&-2&1&0&0&0&0&0\\
       0&1&-2&1&0&0&0&0\\
       -\bshft+&0&1&-1&0&0&0&\bshft-\\
       \bshft+&0&0&0&-1&1&0&-\bshft-\\
       0&0&0&0&1&-2&1&0\\
       0&0&0&0&0&1&-2&1\\
       0&0&0&\bshft+&-\bshft+&0&1&-1
\end{bmatrix}.
\label{eq:ll1}
\end{equation}
In a homotopy from a useful base for centre manifold theory, we must end the homotopy at this operator for the fully coupled dynamics on the elements.

\subsection{A homotopy connects elements}
\label{sec:hce}

Now seek a base of decoupled elements from which a slow manifold may be constructed to model the coarse grid dynamics.
Perhaps the closest operator to the fully coupled~$L_1$ of~\eqref{eq:ll1} is simply to change the inter-element coupling shift operators~$\bshft\pm$ to be ones: that is, define
\begin{equation}
L_0=L+Y_1+Y_3=\begin{bmatrix} -1&1&0&-1&1&0&0&0\\
       1&-2&1&0&0&0&0&0\\
       0&1&-2&1&0&0&0&0\\
       -1&0&1&-1&0&0&0&1\\
       1&0&0&0&-1&1&0&-1\\
       0&0&0&0&1&-2&1&0\\
       0&0&0&0&0&1&-2&1\\
       0&0&0&1&-1&0&1&-1
\end{bmatrix}.
\label{eq:ll0}
\end{equation}
Then straightforward algebra derives that the spectrum of $D\dot{\vec v}_j=L_0\vec v_j$ is $\{0,-2/3,-2,-4\}$ for each of the $m$~decoupled elements.  The zero eigenvalue with all the rest negative implies, Section~\ref{sec:cmtsm}, that there exists a relevant slow manifold model~\cite[e.g.]{Carr81, Kuznetsov95} which we can construct with one dimension for each element---the dynamics on the $m$-dimensional slow manifold forms the coarse grid model.

But is~$L_0$ a good choice? and, is it the only choice?  The spectrum of the homotopy answers.  Create a general homotopy from decoupled elements to fully coupled by defining the convex combination
\begin{displaymath}
L_\gamma = (1-\gamma)\left[ L_0+y_1Y_1+y_2Y_2+y_3Y_3 \right] +\gamma L_1\,,
\end{displaymath}
where parameter~$\gamma$ morphs the operator from the decoupled case, $\gamma=0$\,, to the fully coupled case, $\gamma=1$\,.  The characteristic polynomial, $\det(L_\gamma-\lambda D)$, is too hideous to record here, but computer algebra derives it easily.  Then computer algebra iteration finds asymptotic approximations to the eigenvalue near zero:
\begin{displaymath}
\lambda=\gamma^2\rat1{16}(\bshft-^2-2+\bshft+^2)\big[
1+\rat14(y_1-18y_2-10y_3)\big] +\Ord{\gamma^3+|\vec y|^2}\,.
\end{displaymath}
The factor $\rat1{16}(\bshft-^2-2+\bshft+^2)$, although unexpectedly involving  double shifts over the coarse grid, is exactly what we want as to leading order it is the centred second difference~$\delta^2$ of the fine grid diffusion dynamics~\eqref{eq:dde}.  However, the term linear in~$y_l$ ruins this identity and so we choose $y_1=18y_2+10y_3$\,.  Similar computer algebra to higher order in~$|\vec y|$ indicates we also need $y_3=-3y_2$\,, and analysis to higher order in coupling parameter~$\gamma$ indicates that then we need $y_2=0$\,.  To best match the dynamics of the discrete diffusion~\eqref{eq:dde} we choose $y_1=y_2=y_3=0$\,. Consequently, accuracy of the coarse grid model for just simple diffusion requires us to use the convex combination operator
\begin{equation}
L_\gamma = \gamd L_0 +\gamma L_1\,,
\qtq{where normally} \gamd=1-\gamma\,,
\label{eq:llgam}
\end{equation}
as the homotopy to smoothly connect the decoupled matrix~\eqref{eq:ll0} with the fully coupled operator~\eqref{eq:ll1}.

\paragraph{Inhomogeneity requires generalisation}
Linearisation in Section~\ref{sec:cmtsm} leads us to consider the class of inhomogeneous, discrete, reaction-diffusion equations on the fine grid of
\begin{equation}
\dot u_i=\delta(f_i\delta u_i)+\alpha g_iu_i
=f_{i-\frac12}u_{i-1}
-(f_{i-\frac12}+f_{i+ \frac12})u_i
+f_{i+ \frac12}u_{i+1} -\alpha g_iu_i\,,
\label{eq:idrde}
\end{equation}
for some spatially varying `reactions'~$g_i$ and  `diffusivities'~$f_{i\pm\frac12}$ governing the dynamic exchange between $u_i$~and~$u_{i\pm1}$.  Equation~\eqref{eq:idrde} describes relatively general self-adjoint dynamics on the fine grid. 

Embed the fine grid dynamics~\eqref{eq:idrde} into the dynamics on the finite elements shown in Figure~\ref{fig:oepsad} and generalising the self-adjoint, consistent operator~\eqref{eq:llgam} by considering the dynamics
\begin{equation}
D\dot{\vec v}_j= F_j\vec v_j
+\alpha \vec g_j\,,
\label{eq:edrde}
\end{equation}
where the linear `reaction'
\begin{displaymath}
\vec g_j=\big(0,g_{2j-\frac32}v_{j,1},g_{2j-\frac12}v_{j,2},0
             ,0,g_{2j+\frac12}v_{j,5},g_{2j+\frac32}v_{j,6},0 \big)\,,
\end{displaymath}
and where the diffusivity matrix 
\begin{eqnarray}
F_j&=&\begin{bmatrix}
F_{1,1}&F_{1,2}\\F_{2,1}&F_{2,2}
\end{bmatrix}
\qtq{for the four sub-blocks} \label{eq:edrdef} \\
F_{1,1}&=&\begin{bmatrix}
-f_{2j-2}&+f_{2j-2}&0&-(\gamd +\gamma\bshft-)f_{2j} \\
+f_{2j-2}&-f_{2j-2}-f_{2j-1}&+f_{2j-1}&0 \\
0&+f_{2j-1}&-f_{2j-1}-f_{2j}&+f_{2j} \\
-f_{2j}(\gamd +\gamma\bshft+)&0&+f_{2j}&-f_{2j} \\
\end{bmatrix},
\nonumber\\
F_{2,2}&=&\begin{bmatrix}
-f_{2j}&+f_{2j}&0&-f_{2j}(\gamd +\gamma\bshft-) \\
+f_{2j}&-f_{2j}-f_{2j+1}&+f_{2j+1}&0 \\
0&+f_{2j+1}&-f_{2j+1}-f_{2j+2}&+f_{2j+2} \\
-(\gamd +\gamma\bshft+)f_{2j}&0&+f_{2j+2}&-f_{2j+2} \\
\end{bmatrix},
\nonumber\\
F_{1,2}&=&\diag[ 
(\gamd +\gamma\bshft-)f_{2j}, 0, 0, f_{2j}(\gamd +\gamma\bshft-) ],
\nonumber\\
F_{2,1}&=&\diag[ 
f_{2j}(\gamd +\gamma\bshft+), 0, 0, (\gamd +\gamma\bshft+)f_{2j} ],
\nonumber
\end{eqnarray}
recalling that $\bshft\pm f_{2j}v_{j,i}=f_{2j\pm2}v_{j\pm1,i}$\,.
The operator on the right-hand side of~\eqref{eq:edrde} is self-adjoint:  the only subtlety arises via the operators coupling elements.  For example, suppose the operator $\bshft\pm f_{2j}$ occurs at element~$(i,k)$ in~$F_j$: then the inner product $\inpr{ \vec w}{F\vec v}$ has the components
\begin{eqnarray*}
\sum_j w_{j,i}\bshft\pm (f_{2j}v_{j,k})
&=& \sum_j w_{j,i}f_{2j\pm2}v_{j\pm1,k}
\\&=& \sum_j w_{j\mp1,i}f_{2j}v_{j,k}
\\&=& \sum_j f_{2j}(\bshft\mp w_{j,i})v_{j,k}
\end{eqnarray*}
which suitably corresponds to the $(k,i)$~elements of~$F_j$ in the inner product $\inpr{ F\vec w}{ \vec v}$.  Thus the dynamics of the embedded, coupled, finite element, system~\eqref{eq:edrde} preserves the self-adjointness in the fine grid inhomogeneous dynamics~\eqref{eq:idrde} in a manner consistent with the required homotopy~\eqref{eq:llgam} of homogeneous dynamics.

\subsection{Centre manifold theory supports a coarse grid model}
\label{sec:cmtsm}

We have arrived at a separation of the self-adjoint, isotropic, diffusion dynamics~\eqref{eq:dde} into elements, drawn schematically in Figure~\ref{fig:oepsad}, with coupling between the elements that ranges from decoupled to fully coupled, and preserves the self-adjoint nature of the linear dynamics.  We now invoke centre manifold theory~\cite[e.g.]{Carr81, Kuznetsov95} to rigorously support coarse grid modelling of \emph{nonlinear} lattice dynamics.  The nonlinear dynamics are introduced, embedded into elements, and then linearised to connect to the preceding results.

In analogy with the continuum dynamics of~\eqref{eq:nde}, here we consider the class of fine grid dynamics expressible by the general nonlinear, local interaction, lattice rule
\begin{equation}
\dot u_i=\delta\big[f(i,\mu u_i,\delta u_i)\delta u_i \big] 
+\alpha g(i,u_i, \mu\delta u_i)\,,
\label{eq:ndde}
\end{equation} 
for sufficiently smooth `diffusivities'~$f$ and `reactions'~$g$.  For definite theoretical statements, suppose the fine grid lattice and the dynamics~\eqref{eq:ndde} on it are $2m$-periodic in~$i$.

Embed the fine grid, $2m$-periodic, nonlinear reaction-diffusion dynamics~\eqref{eq:ndde} into the higher dimensional dynamics of~$\vec v_j$ within the corresponding $m$~elements of Figure~\ref{fig:oepsad} by a nonlinear version of~\eqref{eq:edrde}. To help write the map from the fine grid to the overlapping elements let $i'=(i-\rat72)-\operatorname{sign}(i-\rat72)$, then the nonlinear fine grid dynamics~\eqref{eq:ndde} in each element are
\begin{equation}
\dot v_{j,i}=\delta\big[ f_{2j+i'}\delta v_{j,i} \big]
+\alpha g_{2j+i'},
\quad i=1,2,5,6,
\label{eq:ndde1}
\end{equation}
where $f_{2j+i'}=f(2j+i',\mu v_{j,i},\delta v_{j,i})$ and $g_{2j+i'}=g(2j+i',v_{j,i},\mu\delta v_{j,i})$.  
Control the coupling of these to neighbouring elements by the conditions
\begin{eqnarray}
&&f_{2j\pm2}\delta v_{j,i}+(\gamd+\gamma\bshft\pm) 
\big[ f_{2j} \delta v_{j,\frac72} \big]=0\,,
\quad i=\scriptsize\left\{\begin{array}{c}
13/2\\ 1/2 \end{array}\right.,
\\&& f_{2j}(\gamd+\gamma\bshft+)v_{j,0}
+f_{2j}\delta v_{j,i}-f_{2j}(\gamd+\gamma\bshft-)v_{j,7}=0\,,
\quad i=\rat52,\rat72\,.
\label{eq:ndde3}
\end{eqnarray}
The embedded element dynamics~\eqref{eq:ndde1}--\eqref{eq:ndde3} then has a useful $m$~dimensional subspace~$\cE_0$ of equilibria: $\alpha=\gamma=0$ and $v_{j,i}=U_j$ constant in each of the $m$~elements.

Linearise the dynamics about each equilibria in~$\cE_0$ to obtain~\eqref{eq:edrde} without reaction, $\alpha=0$\,, and with diffusivites $f_{2j+i}=f(2j+i,U_j,0)$.  As the coupling parameter $\gamma=0$\,, each element is isolated.  When the diffusivities~$f_i$ are constant, each element has spectrum proportional to $\{0,-2/3,-2,-4\}$ ---the zero eigenvalue corresponds to $v_{j,i}$~being constant in each element. Elementary algebra shows that provided for all~$j$
\begin{equation}
f_j>0 \qtq{and}
2f_{2j}(f_{2j-2}+f_{2j+2})-f_{2j-2}f_{2j+2}>0\,,
\label{eq:prespec}
\end{equation}  
then the spectrum remains as one zero eigenvalue with the other three being negative. Consequently, centre manifold theory assures us of the following corollary~\cite[e.g.]{Carr81, Kuznetsov95}.

\begin{corollary}[slow manifold] \label{thm:cmtd}
Provided~\eqref{eq:prespec}, in some finite neighbourhood of the subspace~$\cE_0$:
\begin{enumerate}
\item there exists a $(m+2)$~dimensional slow manifold~$\cM_0$ of the nonlinear, fine grid, element dynamics~\eqref{eq:ndde1}--\eqref{eq:ndde3}, one dimension for each element, one for the inter-element coupling parameter~$\gamma$, and one for the amplitude parameter~$\alpha$ of the reaction;

\item the slow manifold~$\cM_0$ may be parametrised by any reasonable measure~$U_j$ of~$v_{j,i}$ in each element, that is, the slow manifold and the evolution thereon may be written, for some $\vec v_j$~and~$g_j$, $j=1,\ldots,m$\,,  as 
\begin{equation}
\vec v_j=\vec v_j(\vec U,\gamma,\alpha)
\qtq{such that}
\dot U_j=\frac{dU_j}{dt}=g_j(\vec U,\gamma,\alpha)\,;
\label{eq:smd}
\end{equation}

\item \label{i:reld} the dynamics on~$\cM_0$ is `asymptotically complete'~\cite{Robinson96} in that from all initial conditions in some neighbourhood of~$\cM_0$, there exists a solution of~\eqref{eq:smd} approached exponentially quickly in time by the solution of~\eqref{eq:ndde1}--\eqref{eq:ndde3};

\item \label{i:accd}   the order of error of an \emph{approximation} to the slow manifold~$\cM_0$ \emph{and} its evolution,~\eqref{eq:smd}, is the same as the order of the residuals of the governing dynamics~\eqref{eq:ndde1}--\eqref{eq:ndde3}, in the coupling parameter~$\gamma$ and reaction parameter~$\alpha$, when evaluated at the approximation.
\end{enumerate}  
\end{corollary}

\paragraph{Basic example of linear diffusion}
The simplest example of the class of lattice dynamics to which Corollary~\ref{thm:cmtd} applies is the linear discrete diffusion~\eqref{eq:dde}.
Computer algebra~\cite[\S4]{Roberts08j} readily iterates to asymptotically approximate the slow manifold model.   Here we choose to parametrise the slow manifold of coarse scale dynamics in terms of the coarse variables
\begin{equation}
U_j=\rat14(v_{j,2}+v_{j,3}+v_{j,4}+v_{j,5})\,,
\end{equation}
which Figure~\ref{fig:oepsad} shows to be estimates of the mid-element values of the fine grid variables. Executing the computer algebra deduces that the fine grid, intraelement structure is
\begin{eqnarray*}
\vec v_j&=&\Big[ \phantom{+\gamma} (1,1,1,1,1,1,1,1)
\\&&{}
+\frac\gamma4 (-5,-3,-1,1,-1,1,3,5)\bmu\bdelta
\\&&{}
+\frac{\gamma^2}4 (3,1,0,0,0,0,1,3)(\bdelta^2 +\rat14\bdelta^4)
\\&&{}
+\frac{\gamma^2}{16}(13,7,1,-5,5,-1,-7,-13)\bmu\bdelta^3
\Big] U_j
+\Ord{\gamma^3},
\end{eqnarray*}
in terms of the centred mean~$\bmu$ and difference~$\bdelta$ operators on the coarse grid.  The first line gives the piecewise constant basis for the slow subspace~$\cE_0$ of equilibria.  The second line gives a linear variation between neighbouring elements, the third line quadratic, and so on.  Corollary~\ref{thm:cmtd} assures us that no matter what the initial conditions for the fine grid diffusion~\eqref{eq:dde}, exponentially quickly the system will settle onto the slow manifold with the above structure on the fine grid.

The corresponding evolution of the coarse grid variables~$U_j$ is
\begin{eqnarray}
\dot U_j&=&\Big[ \gamma^2\left( \rat14\bdelta^2+\rat1{16}\bdelta^4 \right)
+\gamma^3\left( -\rat5{16}\bdelta^4 -\rat5{64}\bdelta^6 \right)
\nonumber\\&&{}
+\gamma^4\left( \rat{15}{64}\bdelta^4 +\rat{53}{128}\bdelta^6 +\rat{91}{1024}\bdelta^8 \right) \Big]U_j
+\Ord{\gamma^5}\,.
\label{eq:cgm}
\end{eqnarray}
Obtain the lowest accuracy model from the first line, by neglecting terms $\Ord{\gamma^3}$, and then evaluating at the physical fully coupled case $\gamma=1$\,:  the model is $\dot U_j=\rat1{16}(U_{j-2}-2U_j+U_{j+2})$ which is appropriate although surprisingly only involves every second coarse grid value.  This `surprise' was forecast in Section~\ref{sec:hce} by the leading operator eigenvalue of the operator~$L_\gamma$. Higher orders in coupling parameter~$\gamma$ give coarse grid models with a wider stencil, and of more accuracy.  For example, we see the accuracy through the \emph{equivalent fine grid expression} of the coarse grid model~\eqref{eq:cgm} obtained by the operator identity that $\bdelta^2=4\delta^2+\delta^4$\,: in terms of fine grid differences~$\delta$, the coarse model~\eqref{eq:cgm} is
\begin{displaymath}
\dot U_j=\left[ \gamma^2\delta^2 
+\rat54\gamma^2(1-\gamma)(1-3\gamma)\delta^4 \right]U_j
+\Ord{\delta^6,\gamma^5}\,.
\end{displaymath} 
This equivalent fine grid expression demonstrates that when evaluated for full coupling, $\gamma=1$\,, the fourth order differences vanish to leave the correct equivalent model $\dot U_j\approx\delta^2 U_j$\,. Computer algebra~\cite[\S4]{Roberts08j} to high order in coupling~$\gamma$ confirms that higher order differences in the equivalent fine grid expression similarly vanish.  Thus the coarse grid model~\eqref{eq:cgm} is an accurate closure for the coarse scales of the fine grid dynamics.

As introduced in earlier work with non-self-adjoint coupling~\cite{Roberts08c}, such a mapping of lattice dynamics from fine grid scale to the coarser scale can be iterated and renormalised to cover step-by-step the wide range of length and time scales on a multigrid~\cite[e.g.]{Briggs01,Thomas03}.  Such step-by-step dynamical transformations could empower us, in future research, to carefully explore the development of emergent phenomena in nonlinear and stochastic systems over multigrids.

\section{Three further applications indicates range}
\label{sec:tfair}

The previous section concluded by modelling the dynamics of discrete diffusion.  By itself, discrete diffusion is well understood.  The value of the preceding section is that it empowers us to model more complicated dynamics from the same base with the same powerful centre manifold support.  This section introduces three interesting applications.

\subsection{Reaction-diffusion lattice dynamics}

Centre manifold theory, as recorded in Corollary~\ref{thm:cmtd}, applies to nonlinear dynamics such as the coarse grid modelling of the discrete reaction-diffusion equation~\cite[e.g.]{Wollkind94}
\begin{equation}
\dot u_i=\delta^2u_i +\alpha(u_i-u_i^2)\,,
\label{eq:drde}
\end{equation}
where $\alpha(u_i-u_i^2)$ is some example nonlinear reaction.  

Simple modifications of earlier code~\cite[\S5]{Roberts08j} gives computer algebra that constructs the slow manifold model of the coarse grid evolution.  Executing the code derives, for example, the coarse grid model
\begin{eqnarray}
\dot U_j&=& \frac{\gamma^2}{16}(U_{j+2}-2U_j+U_{j-2})
+\alpha(U_j-U_j^2)
\nonumber\\&&{}
+\frac{\alpha\gamma^2}{64}\big(
14U_j^2-10U_jU_{j+1}^2
-5U_{j+1}^2-4U_jU_{j+2}
+10U_{j+1}U_{j+2}\nonumber\\&&\quad{}
-3U_{j+2}^2-10U_jU_{j-1}
+10U_{j-1}U_{j+1}-5U_{j-1}^2-4U_jU_{j-2}
\nonumber\\&&\quad{}
+10U_{j-1}U_{j-2}-3U_{j-2}^2 \big)
+\Ord{\gamma^3,\alpha^2}.
\end{eqnarray}
The first line, at full coupling $\gamma=1$, is a classic model of the reaction-diffusion~\eqref{eq:drde}.  The second and subsequent lines account for sub-element reaction and diffusion, interacting together and with neighbouring elements.  Such terms are required for an accurate closure of the nonlinear fine grid dynamics on the coarse grid.

\subsection{Homogenisation}
\label{sec:h}

A critical issue in material science is the effective large scale properties of a composite material with significant microscopic structure.  A canonical problem is the effective large scale diffusion through a domain with microscopic variations in diffusion coefficient~\cite[e.g.]{Brandt00, Samaey03b, Arbogast06}.  Here we transform diffusion on a fine grid, with fine grid variations in coefficient, into diffusion onto a coarser grid.  Here we provide a new and powerful view of the classic result that, to leading order, the coarse grid diffusion is a local average of the fine grid diffusion.

Consider the fine grid dynamics~\eqref{eq:ndde}, without reaction $\alpha=0$\,, where the diffusion coefficient governing flux between fine grid points $x_{i-\frac12}$~and~$x_{i+\frac12}$ is $f_i=1+\epsilon\kappa_i$\,. The parameter~$\epsilon$ controls the overall size of the variations in the diffusivity.  For simplicity in construction and interpretation I treat the variations in diffusivity as small; that is, we construct the slow manifold as a power series in~$\epsilon$. Computer algebra then finds the slow manifold as a power series in the strength~$\epsilon$.  Straightforward modifications to earlier computer algebra~\cite[\S6]{Roberts08j} finds, for the example~\eqref{eq:introhomo} mentioned in the Introduction,
\begin{equation}
\dot U_j=\gamma^2\rat1{16}\big\{{\cal K}_{j-1}U_{j-2}
-({\cal K}_{j-1}+{\cal K}_{j+1})U_j
+{\cal K}_{j+1}U_{j+2} \big\} +\Ord{\gamma^3} \,,
\label{eq:cgddeh}
\end{equation}
where the coarse grid effective diffusivity coefficients
\begin{eqnarray*}
{\cal K}_j&=&
1+\epsilon\rat14\big[ \kappa_{2j-2} +\kappa_{2j-1} +\kappa_{2j+1} +\kappa_{2j+2}\big]
\\&&{}
+\epsilon^2\rat1{16}\big[
-(\kappa_{2j-2}+\kappa_{2j-1}-\kappa_{2j+1}-\kappa_{2j+2})^2
\\&&\quad{}
-2(\kappa_{2j-2}-\kappa_{2j-1})^2
-2(\kappa_{2j+2}-\kappa_{2j+1})^2
\big] +\Ord{\epsilon^3}\,.
\end{eqnarray*}
Because the element coupling preserves self-adjoint symmetry we find that the coarse grid, slow manifold model~\eqref{eq:cgddeh} is indeed self-adjoint with these particular effective diffusivities.

Even more beautiful is that the effective diffusivity on the coarse grid is \emph{local}: the diffusivity~${\cal K}_{j\pm1}$ governing the flux between coarse grid points $X_j$~and~$X_{j\pm2}$ depends only upon the fine scale diffusivities between $X_j$~and~$X_{j\pm2}$, namely between $\kappa_{2j}$~and~$\kappa_{2j\pm4}$. This beautiful feature is not built into the approach, but appears as a natural consequence of this scheme to preserve self-adjoint properties.

\paragraph{Zigzag microstructure}
The specific zigzag microstructure $\kappa_i=(-1)^i$ is straightforward to analyse to higher order~\cite[\S6]{Roberts08j}, and the resultant macroscale model compact enough to record. With zigzag microscale diffusivity, the diffusivity is uniform on the coarse grid: writing the coarse scale evolution as $\dot U_j={\cal D}U_j$\,, the coarse grid diffusion operator
\begin{eqnarray}
{\cal D}&=&\rat14\gamma^2(1-\epsilon^2)(1+\rat1{4}\bdelta^2)\bdelta^2
\Big\{ 1
-\rat1{4}\gamma(5+\epsilon)\bdelta^2
\nonumber\\&&{}
+\rat1{64}\gamma^2\big[60+16\epsilon+4\epsilon^2 
+(91+36\epsilon+5\epsilon^2)\bdelta^2\big]
\bdelta^2 \Big\}
\nonumber\\&&{}+\Ord{\gamma^5,\epsilon^5}\,.
\end{eqnarray}
As it should, the $1-\epsilon^2$ factor shows that the coarse scale diffusion is depressed by the microstructure, and indeed vanishes for $\epsilon=\pm1$  reflecting the vanishing of the microscale diffusivity between every second fine scale grid point that occurs at $\epsilon=\pm 1$\,.

\subsection{Pattern evolution}

An outstanding issue in modelling is the direct discrete, macroscale modelling of pattern evolution~\cite[e.g.]{Cross93,Brand88}.  Currently, in fields such as fluid convection and in some reaction-diffusion equations, we first model the evolution of rolls, spirals and spots by variants of the Ginzburg--Landau \pde\ derived assuming amplitudes and phases of the structures vary slowly in space.  Second, we then discretise the Ginzburg--Landau \pde\  for numerical simulations of the modulation of the pattern over large scales.  The challenge is to analyse the original system dynamics and \emph{directly} generate such a macroscale discretisation in \emph{one step}.  My first attempt to do this showed potential~\cite{Roberts01d}, but the coupling conditions used therein did not preserve the self-adjoint nature of the dynamics and so the discretisation unsatisfactorily did not preserve required symmetries.  Now with self-adjoint coupling conditions we return to modelling pattern evolution. 

\paragraph{Linear dynamics} On a lattice, the simplest microscale pattern is perhaps the long lasting, zigzag mode of the discrete equation
\begin{equation}
\dot u_i=-4\mu^2 u_i=-u_{i-1}-2u_i-u_{i+1}\,.
\label{eq:dzze}
\end{equation}
Now adapt the self-adjoint coupling conditions and analysis of Section~\ref{sec:oepsad} to derive a model of these zigzag mode dynamics on the coarser grid.  Embed the dynamics on the elements shown in Figure~\ref{fig:oepsad} with the evolution rule
\begin{equation}
D\dot{\vec v}_j=Z_\gamma\vec v_j\,,
\end{equation}
where the `zigzag' operator
\begin{displaymath}
Z_\gamma=\begin{bmatrix}
-1&-1&0&\gamd +\gamma\bshft-&\gamd +\gamma\bshft-&0&0&0\\
-1&-2&-1&0&0&0&0&0\\
0&-1&-2&-1&0&0&0&0\\
\gamd +\gamma\bshft+&0&-1&-1&0&0&0&\gamd +\gamma\bshft-\\
\gamd +\gamma\bshft+&0&0&0&-1&-1&0&\gamd +\gamma\bshft-\\
0&0&0&0&-1&-2&-1&0\\
0&0&0&0&0&-1&-2&-1\\
0&0&0&\gamd +\gamma\bshft+&\gamd +\gamma\bshft+&0&-1&-1
\end{bmatrix}.
\end{displaymath}
Then based about the decoupled case of $\gamma=0$\,, centre manifold theory~\cite[e.g.]{Carr81, Kuznetsov95} similarly guarantees the existence and relevance of a slow manifold parametrised by the amplitude of the zigzag mode \emph{local} to each element.
Computer algebra~\cite[\S7]{Roberts08j}, modified from earlier code via spatial patterns changed to~$(-1)^i$ to account for the zigzag neutral mode, constructs the slow manifold model. 
The corresponding evolution of the coarse grid order parameters, the amplitudes~$U_j$, is exactly~\eqref{eq:cgm} discussed before.   The difference is that here the equation governs the local amplitude of the zigzag mode, rather than the local mean field.

\paragraph{Nonlinear amplitude modulation}
Even more interesting is nonlinear lattice dynamics which centre manifold theory also supports.  Modify the discrete equation~\eqref{eq:dzze} by a local quadratic nonlinearity to
\begin{equation}
\dot u_i=-4\mu^2 u_i+au_i^2=-u_{i-1}-2u_i-u_{i+1}+au_i^2\,.
\label{eq:dzzen}
\end{equation}
What is the corresponding coarser grid equation governing the local amplitude of the zigzag mode?  In particular, does the quadratic nonlinearity stabilise or destabilise the origin?

Simple modifications to the computer algebra then derives the coarser scale, \emph{modulation} equation
\begin{equation}
\dot U_j=\gamma^2(\rat14\bdelta^2+\rat1{16}\bdelta^4) U_j 
-a\gamma\rat14(\bdelta\bmu U_j)(\bdelta^2U_j)
+\rat12a^2U_j^3
+\Ord{\gamma^2+a^3}\,.
\label{eq:dgle}
\end{equation}
Evaluated at $\gamma=1$ this coarse grid model predicts that the quadratic nonlinearity in~\eqref{eq:dzzen} is destabilising as it generates the cubic growth term~$\rat12a^2U_j^3$\,.  In other problems, the coarse grid model will similarly form a discrete version of the Ginzburg--Landau equation.  One might derive such a Ginzburg--Landau equation by traditional modulation theory, but here  analysis of the fine scale grid dynamics generates the novel term $a\gamma\rat14(\bdelta\bmu U_j)(\bdelta^2U_j)$ which accounts for interactions over the scale of a few grid points and thus forms a more comprehensive closure.

This analysis empowers us to derive discrete models of nonlinear lattice pattern dynamics without having to assume the infinite scale separation required by traditional modulation theory.

\section{Conclusion}

This article uses centre manifold  theory to further develop a \emph{novel} approach to high quality coarse scale discrete models of nonlinear spatiotemporal systems.  The method is to divide space into \emph{overlapping} finite elements with specially crafted coupling conditions.  The innovative coupling conditions of Section~\ref{sec:sacc} not only engender centre manifold support, Section~\ref{sec:cmtsd}, and assure consistency for vanishing element size, but also preserve the self-adjoint symmetry which is often so important in applications.

The first half of this article extracts accurate finite scale discrete lattice models from analysis of the infinitesimal scale modelling of dissipative \pde{}s.  A companion problem is to extract an accurate coarse scale lattice model from a fine scale lattice model.  The second half of this article does this for the possibly extreme case when the coarse scale lattice is just a factor of two coarser than the fine scale lattice, and thus connects this approach to multigrid methods~\cite{Briggs01}.   The approach is again based upon dividing the fine lattice into overlapping elements with specially crafted coupling conditions, Section~\ref{sec:oepsad}. To supplement an earlier approach~\cite{Roberts08c}, here the coupling conditions preserve self-adjointness in the dynamics as seen in the three example applications of Section~\ref{sec:tfair}.

Importantly, centre manifold theory~\cite[e.g.]{Carr81, Kuznetsov95} supports modelling for finite spectral gap, not just the infinite spectral gap required by other methods.  Furthermore, the support applies, in principle, to a finite domain and for all time; the only approximation is in the approximation of the slow manifold. The approach developed here for deterministic dynamics in one spatial dimension should be extendable to both higher dimensions and stochastic dynamics.

\paragraph{Acknowledgement} I thank the Australian Research Council for support via grants DP0774311 and DP0988738.

\bibliographystyle{plain}
\bibliography{ajr,bib}

\end{document}